%% file: main.tex
\documentclass[onecolumn,11pt,notitlepage,nofootinbib,superscriptaddress,numerical]{revtex4-2}

\input{preamble}

\begin{document}

\title{\Large Extractors: QLDPC Architectures for Efficient Pauli-Based Computation}

\author{Zhiyang He (Sunny)\textsuperscript{1,\textdagger}}
\author{Alexander Cowtan\textsuperscript{2}}
\author{Dominic J.~Williamson\textsuperscript{3}}
\author{Theodore J.~Yoder\textsuperscript{4}}
\renewcommand{\thefootnote}{\fnsymbol{footnote}}
\footnotetext[2]{Email: szhe@mit.edu.}
\setcounter{footnote}{0}
\renewcommand{\thefootnote}{\arabic{footnote}}
\footnotetext[1]{Department of Mathematics, Massachusetts Institute of Technology, Cambridge, MA 02139, USA}
\footnotetext[2]{Department of Computer Science, University of Oxford, Wolfson Building, Parks Road, Oxford OX1 3QD, UK}
\footnotetext[3]{IBM Quantum, IBM Almaden Research Center, San Jose, CA 95120, USA}
\footnotetext[4]{IBM Quantum, IBM T.J. Watson Research Center, Yorktown Heights, NY 10598, USA}

\date{\vspace{\baselineskip} \today}

\begin{abstract}
In pursuit of large-scale fault-tolerant quantum computation, quantum low-density parity-check (LDPC) codes have been established as promising candidates for low-overhead memory when compared to conventional approaches based on surface codes. Performing fault-tolerant logical computation on QLDPC memory, however, has been a long standing challenge in theory and in practice. In this work, we propose a new primitive, which we call an $\textit{extractor system}$, that can augment any QLDPC memory into a computational block well-suited for Pauli-based computation. In particular, any logical Pauli operator supported on the memory can be fault-tolerantly measured in one logical cycle, consisting of $O(d)$ physical syndrome measurement cycles, without rearranging qubit connectivity. We further propose a fixed-connectivity, LDPC architecture built by connecting many extractor-augmented computational (EAC) blocks with bridge systems. When combined with any user-defined source of high fidelity $\ket{T}$ states, our architecture can implement universal quantum circuits via parallel logical measurements, such that all single-block Clifford gates are compiled away. The size of an extractor on an $n$ qubit code is $\tilde{O}(n)$, where the precise overhead has immense room for practical optimizations.
\end{abstract}

\maketitle

Quantum error correction~\cite{shor1995scheme,gottesman1997stabilizer,Kitaev1997qec} has been established as a fundamental building block of large-scale, fault-tolerant quantum computation. 
For more than two decades, the surface code~\cite{bravyi1998codes,Kitaev2003anyon,dennis2002memory} has been the leading candidate for practical implementation, due to its plethora of desirable properties. 
Notably, the surface code can be implemented on a two-dimensional lattice of physical qubits with nearest-neighbor connectivity, and achieves the best asymptotic parameters under such connectivity constraints~\cite{bravyi2010tradeoffs}.
Following years of extensive research, the surface code now has fast and accurate decoders~\cite{Delfosse2021almostlineartime,Higgott2025sparseblossom,Wu2023fusion,Bausch2024learning},
practical schemes for logical computation~\cite{Moussa2016transversal,Horsman2012LatticeSurgery,Litinski2018latticesurgery,Chamberland2022universal,gidney2024magic}, 
architectural proposals~\cite{fowler2012surface,Fowler2018lattice,litinski2019game,litinski2022active}, detailed cost analysis~\cite{gidney2021factoring}, and recent milestone demonstrations of subthreshold scaling~\cite{google2023suppressing,acharya2024quantumerrorcorrectionsurface}. We refer readers to Ref.~\cite{eczoo_surface} for more references and expositions.

A critical limitation of the surface code, nonetheless, is its low encoding rate which incurs a significant space overhead in practical and theoretical fault-tolerance. 
This limitation motivated the study of more space-efficient codes, notably quantum low-density parity-check (LDPC) codes.
QLDPC codes relax the constraint on qubit connectivity from 2D nearest-neighbor to arbitrary constant-degree connections, and as a consequence they can have up to constant encoding rate and relative distance~\cite{tillich2013quantum,leverrier2015quantum,breuckmann2017hyperbolic,breuckmann2021quantum,panteleev2021degenerate,hastings2021fiber,panteleev2021quantum,breuckmann2021balanced,panteleev2022asymptotically,leverrier2022quantum,dinur2023good}.
Alongside the theoretical developments in asymptotic code constructions, recent works have proposed QLDPC codes with competitive practical parameters~\cite{panteleev2021degenerate,bravyi2024high,xu2024constant,lin2024quantum,scruby2024high,malcolm2025computing} and memory performances~\cite{bravyi2024high,xu2024constant}.
In parallel, there have been significant advancements in the design of hardware platforms with flexible qubit connectivity~\cite{bluvstein2024logical,atom-reichardt2024logical,quantinuum-paetznick2024demonstration,rodriguez2024experimental,radnaev2024universal}, and long-range connections for hardware with fixed connectivity~\cite{Bravyi2022future,psiquantum2025manufacturable,xanadu-aghaee2025scaling}.
In light of this progress, QLDPC memories have emerged as promising alternatives to surface code memories.

Performing logical computation on QLDPC memory, however, has been a long standing challenge in theory and in practice. 
On a high level, the difficulty arises from the fact that a QLDPC code encodes many logical qubits in the same block, making it generally hard to address individual logical qubits. 
As a result, existing schemes for computation on QLDPC codes often have one or more of the following limitations: they are only applicable to specific codes, have limited action on the logical space, or incur heavy overheads.
We sample the literature to illustrate this barrier. 
Prior works have identified constant-depth gates on different families of QLDPC codes~\cite{breuckmann2024foldtransversal,bravyi2024highthreshold,quintavalle2023partitioning,eberhardt2024operators,zhu2023gates,scruby2024quantum,breuckmann2024cups,hsin2024classifying,lin2024transversal,golowich2024quantum},
which perform various subsets of logical operations. 
Recent work~\cite{malcolm2025computing} constructed a family of codes with low-depth logical Clifford gates, at the expense of inverse-exponential rate.
Besides constant-depth gates, specialized code deformation methods~\cite{bombin2009quantum,vuillot2022quantum} have been used to design addressable logical action on high rate codes~\cite{breuckmann2017hyperbolic,Lavasani2018low,JochymOConnor2019faulttolerantgates,krishna2021fault} with varying degrees of flexibility and overhead cost. 
One specific code deformation technique is QLDPC code surgery, which was first introduced by Ref.~\cite{cohen2022low} as a generalization of lattice surgery~\cite{Horsman2012LatticeSurgery}, building on methods from quantum weight reduction~\cite{hastings2016weight,hastings2021weight,sabo2024weight}. Surgery is a technique that enables flexible measurement of logical operators on any QLDPC code, and its overhead has been significantly improved in the past year~\cite{cowtan2024css,cowtan2024ssip,cross2024improved, swaroop2024universal, williamson2024low,ide2024fault,zhang2024time, cowtan2025parallel}. However, outside of some promising intermediate scale examples~\cite{stein2024architectures}, the overhead of performing logical circuits (rather than individual gates) with surgery is not well-understood.
A similar yet distinct technique is homomorphic measurements~\cite{huang2023homomorphic}, which has been specialized to perform flexible and parallel measurements on homological product codes, enabling several algorithmic primitives~\cite{xu2024fast}.
For full computational proposals, many fault-tolerant schemes based on constant-rate QLDPC codes~\cite{gottesman2013fault,fawzi2020constant,tamiya2024polylog,nguyen2024quantum,he2025composable} perform computation by full-block gate teleportation~\cite{gottesman1999demonstrating,Knill2005},
where gates are teleported one at a time into the memory using distilled Clifford and magic resource states.
While the asymptotic overhead of this approach can be low~\cite{nguyen2024quantum}, it is not currently a practical approach.
In summary, the many current schemes for logical computation on QLDPC codes resemble
disjoint puzzle pieces that are difficult to integrate into a complete picture.

\begin{figure}[t]
    \centering
    \includegraphics[width = 0.62\textwidth]{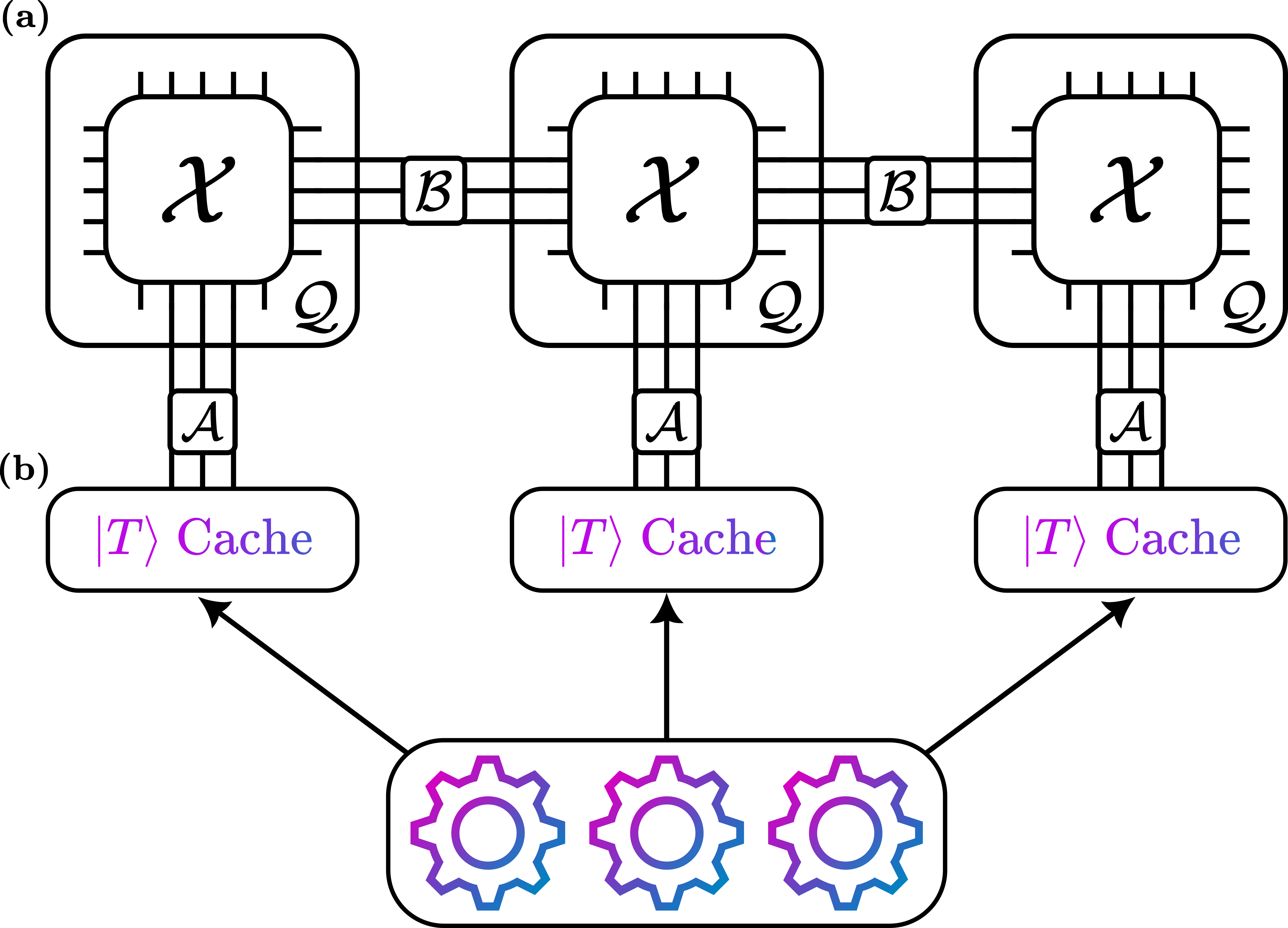}
    \caption{High level depiction of an extractor architecture paired with a magic state factory. 
    \textbf{(a)} Extractor-augmented computational (EAC) blocks $\fullsystem$ connected by bridges $\MB$. In our architecture, these EAC blocks store and operate on logical information via logical Pauli measurements.
    \textbf{(b)} A magic state factory (colored gears) supplying high-fidelity $\ket{T}$ magic states to individual EAC blocks. The output magic states may be stored in local caches, which are connected via adapters $\MA$ to the EAC blocks. If the caches are themselves high-rate QLDPC memories, they can also be equipped with extractors (not drawn) to facilitate the storage and consumption of magic states.
    }
    \label{fig:main_archi}
\end{figure}

\section{Main Results}

Here, we assemble several new and existing ideas to complete a picture of a universal fault-tolerant quantum computer based on QLDPC codes, depicted schematically in Fig.~\ref{fig:main_archi}. 
Our proposed computer consists of several computational blocks, each of which contains a QLDPC code memory $\qcode$ and a novel processing unit $\EXT$ with size proportional to the memory, up to polylog factors. 
The processing unit, which we call an extractor, enables the fault-tolerant measurement of any logical Pauli operator on the memory via QLDPC surgery. 
These computational blocks are further connected together by bridge systems to enable arbitrary joint logical Pauli measurements between memories. 
They are also connected via adapters to magic state factories to enable universal computation. 
Within and between computational blocks and magic state factories, the qubit connectivity is both fixed — unchanging as the computation occurs — and constant-degree — each qubit is only allowed to interact with a small number of other qubits, independent of the computer’s overall size.

Remarkably, such a fixed-connectivity, constant-degree quantum computer can be constructed using any QLDPC code as memory (the choice of memory is even permitted to differ between computational blocks) and any magic state factory that has fixed, constant-degree connectivity. In particular, one can increase the code distances to reduce the logical error rate, provided a subthreshold physical error rate. Because our construction is so flexible, future advances in QLDPC coding theory or magic state production can easily be incorporated to improve the computer’s performance. If the codeblocks possess logical automorphism gates, the fixed-connectivity design can be reduced in size, such that universal fault-tolerant quantum computation can be achieved in combination with those automorphism gates \cite{cross2024improved}. Likewise, the precise connectivity between computational blocks and factories may be chosen with near arbitrary flexibility to satisfy hardware constraints or to simplify the compilation of algorithms and applications. As a result of its flexibility, our construction grants immense freedom for practical optimizations.

\subsection{The Extractor Architecture}

The basis of our architecture is a new computational primitive called an \textbf{extractor} system, built with QLDPC code surgery techniques~\cite{cross2024improved,williamson2024low,swaroop2024universal}, which can augment any QLDPC memory into a computational block.
Specifically, for any $[[n,k,d]]$ QLDPC code $\qcode$, we can build an ancillary extractor system $\EXT$ of data and check qubits and connect it to $\qcode$ to form a fixed, constant-degree system that we call an extractor-augmented computational (EAC) block, denoted $\fullsystem$. 
In an EAC block, the $\qcode$ subsystem holds $k-1$ active logical qubits, the last logical qubit being reserved as an ancilla for computation.
Any logical Pauli operator supported on the $k$ logical qubits in memory can be measured fault-tolerantly in one logical cycle, which involves $O(d)$ physical syndrome measurement cycles, by activating different parts of the extractor system. 
The measurements are performed through a code-switching protocol between $\qcode$ and a QLDPC measurement code $\bar{\qcode}$ supported on $\fullsystem$.
We show that a fixed, constant-degree extractor $\EXT$ can always be built with $O(n(\log n)^3)$ physical qubits.
However, we note this loose theoretical upper bound is unlikely to capture the real overhead one would obtain through practical optimizations on specific code families, which we expect to be a small multiplicative constant. This expectation is supported by existing optimizations of QLDPC surgery on small codes~\cite{cowtan2024ssip,cross2024improved,williamson2024low,ide2024fault}.

To compute on several EAC blocks at once, we use an existing primitive, the bridge/adapter ancilla system developed progressively in Refs.~\cite{cross2024improved,swaroop2024universal}. For two EAC blocks of distance $d$, a bridge/adapter is a fixed, constant-degree ancilla system of $d$ qubits and $d-1$ checks that connect the extractors of the blocks together into a larger extractor.
As a result, the joined EAC blocks behave as a larger EAC block: any logical Pauli operator supported on the $2k$ logical qubits can be measured fault-tolerantly in one logical cycle.

The ability to measure any logical Pauli operator makes EAC blocks an ideal fit for Pauli-based computation (PBC)~\cite{bravyi2016trading}, which compiles an arbitrary Clifford plus $\T$ circuit into a circuit composed of only two primitives: Pauli measurements and $\ket{T}$ magic state preparation.
Consequently, when supplied with high fidelity magic states an EAC block can perform universal computation on the full logical space of $\qcode$. The same bridge/adapter construction used to connect a pair of EAC blocks together can also be used to connect an EAC block to a magic state factory.
This extends our Pauli measurements onto ancillary logical magic states so that they can be consumed for universal computation.

Two important remarks are in order. First, the two EAC blocks $\system{1},\system{2}$ that we connect via a bridge/adapter can be totally different. In particular, $\qcode_1,\qcode_2$ can be arbitrary, different QLDPC codes. 
For this reason, when the two EAC blocks are based on the same code family we call the connecting system a \textbf{bridge} $\mathcal{B}$, and when they are based on different codes, or when we are connecting an EAC to a source of magic states, we call the connecting system an \textbf{adapter} $\mathcal{A}$.
This versatility grants us great flexibility in designing an architecture. In particular, we can connect EAC blocks to any user-defined magic state factory, such as those proposed in Refs.~\cite{bravyi2005magic,Meier2013magic,Jones2013multilevel,OGorman2017magic,litinski2019magic,Gidney2019efficientmagicstate,gidney2024magic,Bravyi2012magic,krishna2018magic,wills2024constant,nguyen2024quantum}, to realize universal fault-tolerant computing. 

We further note that the bridge/adapter can be applied repeatedly to connect many EAC blocks together.  Specifically, an extractor architecture $\archi$ is defined by a graph $\map = (\blocks,\bridges)$, which we call the \textbf{block map}. Every vertex in $\blocks$ corresponds to an EAC block, and every edge in $\bridges$ corresponds to a bridge or adapter system connecting two EAC blocks.

\subsection{Compilation of Universal Quantum Circuits}
\label{sec:main_compilation}

For computation, we allocate one logical qubit per EAC block to serve as an ancilla. Let $B = |\blocks|$ be the number of EAC blocks, then our logical workspace has $K := B(k-1)$ qubits. 
To fault-tolerantly execute a Clifford plus $\T$ circuit $C$ on $K$ qubits, we need to first partition the qubits into the EAC blocks. 
Given a partition $\pt$, we say that a $\CNOT$ gate in $C$ is \textit{in-block} if both of its target qubits belong to the same EAC block, and \textit{cross-block} if they belong to different EAC blocks.
Since the block map $\map$ is user-defined, we leave the choice and optimization of $\pt$ and $C$ to the user as well and simply ask that every cross-block $\CNOT$ gate in $C$ acts on two EAC blocks that are connected by a bridge (equivalently, an edge in $\bridges$).
Note that as long as $\map$ is a connected graph, any circuit $C$ can be compiled (with $\mathsf{SWAP}$ gates, for instance) on any partition $\pt$. 

Our compilation scheme owes its inspiration to the work of Litinski~\cite{litinski2019game}.
We outline the procedure and refer readers to Section~\ref{sec:compilation} and Figure~\ref{fig:compilation} for more details.
The gates in $C$ can be grouped into three types: $\T$ gates, cross-block $\CNOT$ gates, and in-block Clifford and Pauli gates. 
Every $\T$ gate is a $\zt$ rotation, while every $\CNOT$ gate can be implemented as one $\zx$ rotation, followed by two single-qubit Clifford gates. 
The first step of our compilation is to translate all $\T$ gates and cross-block $\CNOT$ gates into their respective Pauli rotations. After this step, the circuit is composed of two types of operations: single-qubit $\pi/8$ and two-qubit $\pi/4$ Pauli rotations (with respect to $\pt$), and in-block Clifford and Pauli gates.

Since Clifford operators permute the Pauli group, Pauli rotations $P_{\theta}$ can be exchanged with a Clifford operator $U$ up to conjugation into another Pauli rotation $(U^{\dagger}PU)_{\theta}$. 
The remaining Clifford gates in $C$ are all in-block, which means if we exchange them with the Pauli rotations, the conjugated Pauli rotations still have the same block support. 
Therefore, we can compile into the following form: single-block $\pi/8$ and two-block $\pi/4$ Pauli rotations, followed by in-block Clifford operators. 
We now make the simplifying assumption that $C$ ends with a round of standard computational basis measurement on all qubits. Then all in-block Clifford operators can be absorbed into the last round of measurements, turning a $Z$ measurement on a single qubit into a single-block Pauli measurement. 
Finally, every $\pi/4$ and $\pi/8$ Pauli rotation can be implemented by two Pauli measurements, followed by controlled Clifford or Pauli corrections.
This finishes our compilation of $C$ into the resulting circuit $\compiledC$.

We now discuss the time overhead of our execution of $\compiledC$. 
A critical feature of an EAC block is that any logical Pauli operator, regardless of its physical and logical support, can be measured in one logical cycle.
Since the circuit is now composed of single- and two-block measurements, its execution can be highly parallelized. 
In particular, any collection of Pauli measurements supported on disjoint blocks can be measured in parallel, as we assumed that any two-block measurements will be supported on EAC blocks connected by bridges. 
We therefore need to solve a scheduling problem on the circuit $\compiledC$.

For the input Clifford plus $\T$ circuit $C$, let $\newc$ denote the circuit we obtain by removing all in-block Clifford gates via the process described above.\footnote{Note that $\newc$ is very different from the compiled circuit $\compiledC$.}
Let $\depth$ be the depth of $\newc$, we say that $C$ has \textit{reduced depth} $\depth$.
We show that 
\begin{equation}\label{eq:main_time}
    \text{Depth of $\compiledC$} < 4k\cdot \depth + k.
\end{equation}
This bounds the number of logical cycles needed to execute all measurements in $\compiledC$. 
Many of these measurements will be supported on magic states, 
and the time cost of supplying these magic states to the EAC blocks needs to be accounted for on specific instantiations of this architecture. 
We discuss several cases in Section~\ref{sec:compilation} of Methods.  
We also note that the above equation is a loose and simplified upper bound on the execution depth, as it does not account for the choice and optimization of $\pt$ and $C$.
We leave more accurate resource estimation of specific algorithms, such as factoring~\cite{Shor}, with optimized choices of $C, \qcode, \map$, $\pt$ and factory, to future works.

\begin{figure}[t]
    \centering
    \includegraphics[width = 1\textwidth]{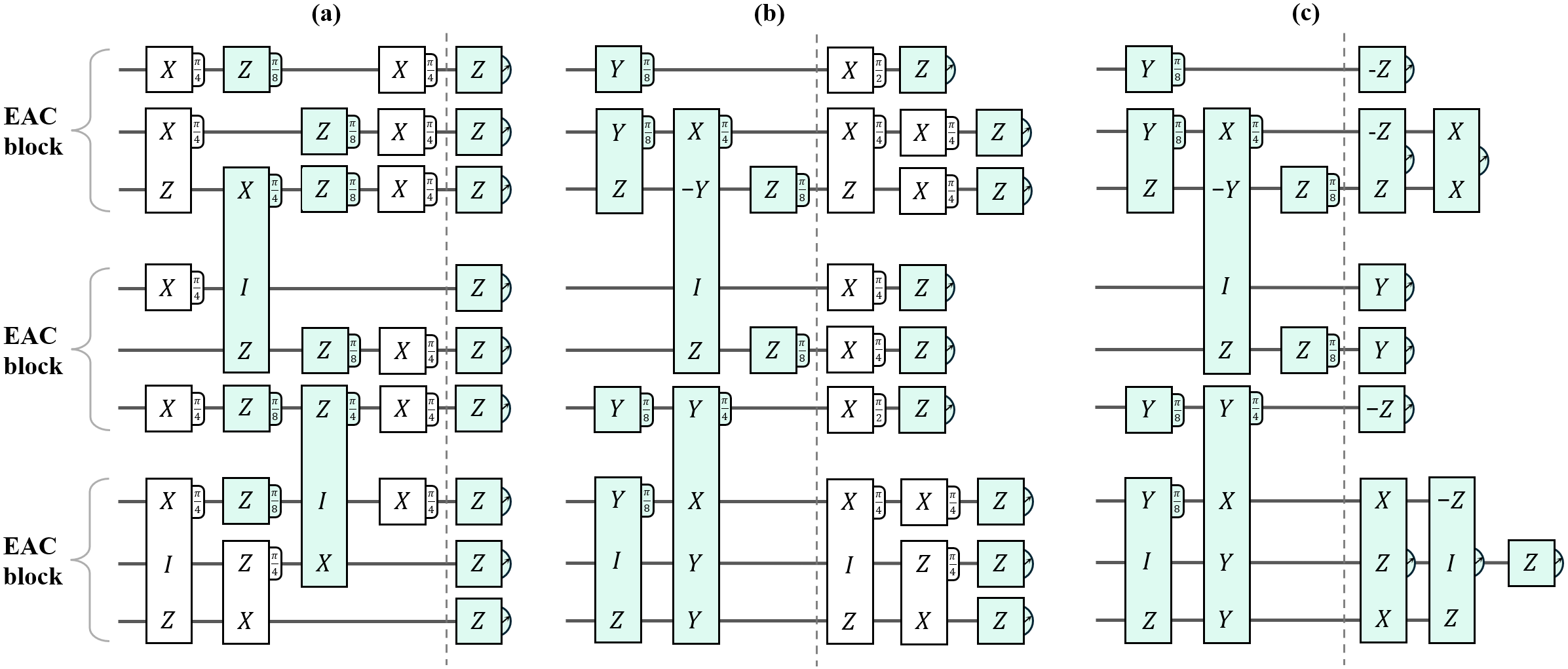}
    \caption{
    Example of circuit compilation for our extractor architecture. 
    The colored components are the ones that we compile and execute on EAC blocks, and the uncolored ones are compiled away.
    \textbf{(a)} A circuit composed of Pauli $\pi/4$ rotations (which are Clifford gates), $\zt$ rotations, and $Z$ measurements at the end. Qubits are grouped into EAC blocks of size 3. 
    \textbf{(b)} Exchanging all $\pi/8$ rotations and cross-block $\pi/4$ rotations to the beginning of the circuit, and all in-block Clifford gates to the end.
    \textbf{(c)} Absorbing all in-block Clifford gates into the final measurements.
    The remaining Pauli rotations will be implemented by Pauli measurements, as depicted in Figure~\ref{fig:compilation} of Methods.
    }
    \label{fig:main_circuit}
\end{figure}

\section{Discussions and Open Questions}\label{sec:main_discussions}

In this work, we design the extractor primitive and propose a highly customizable, low-overhead QLDPC architecture which, when paired with any user-defined source of magic states, can execute any universal quantum circuit supported on its workspace. 
The computational capacity of EAC blocks enables us to compile all in-block Clifford gates away, therefore the execution time overhead is parameterized by magic state throughput and the reduced depth of $C$.
This theoretical blueprint opens a variety of practical and theoretical questions for consideration. 
We survey and discuss a few of them here. 

The most important task ahead is to derive detailed resource estimation for running specific algorithms with optimized choices of circuit $\C$, code $\qcode$, extractors $\EXT$, architectural layout (block map $\map$ and magic state factories), and compilation. 
To this end, there are several open directions.

\paragraph{Optimization of extractor systems.} In this work, we give a construction of extractors in $\tilde{O}(n)$ qubits for any code $\qcode$. However, significant optimizations can be performed on specific codes, as illustrated by the current design on the $[[144,12,12]]$ bivariate bicycle code~\cite{cross2024improved,yoder2016universal}.
Therefore, it is important to design concrete extractor systems for specific QLDPC codes (with promising practical parameters) and optimize for minimal space, connectivity and time overhead (syndrome extraction circuits and decoders).
For codes with constant depth logical Clifford gates, the overhead may be further reduced by constructing a partial extractor, as we discuss in Section~\ref{sec:partial_extractors} of Methods.
This is a vast open area for future exploration. We discuss many practical considerations regarding extractor constructions in Section~\ref{sec:toolkit_practical} and~\ref{sec:extractor_practical} of Methods.

\paragraph{Architectural choices.}
Given optimized constructions of EAC blocks, there are many choices to be made in architectural design. 
Specifically, the choice of $\map$ and magic state factories is highly hardware and application dependent. 
For hardwares with flexible connectivity, such as neutral atom devices, different algorithms and applications can be implemented with different extractor architectures.
Note that the architecture we consider in this paper is uniform, in the sense that all EAC blocks are the same. 
An interesting variant would be a hybrid architecture, where EAC blocks based on different QLDPC codes are connected via adapters. 
Such a hybrid design enables us to optimize different parts of an architecture for different applications or circuit components.

\paragraph{Compilation and time overhead.}
With a fixed architecture and algorithm in mind, the choice of circuit $\C$ and partition $\pt$ is pivotal to the implementation efficiency.
The compilation and scheduling we have analyzed yield loose upper bounds that leave vast room for practical optimizations. 
In this work, we made the simplifying assumption that every cross-block $\CNOT$ gate is supported on two EAC blocks connected by a bridge system. 
This allows us to partition the block map $\map$ into edges for parallel measurements. 
In general, we can partition $\map$ into vertex-disjoint connected subgraphs, since extractors and bridges allow us to perform a joint logical Pauli measurement across all the EAC blocks in a connected subgraph in one logical cycle. 
This grants us more flexibility in measurements and therefore our choice of $\C$.
Note that by grouping EAC blocks together or simply choosing larger EAC blocks, we can compile more Clifford gates away at the expense of less parallel magic state teleportation. 

At the moment, an EAC block measures one logical Pauli operator per logical cycle.
It is interesting to consider whether one could design extractor systems that enable simultaneous measurements of commuting Pauli operators, similar to the protocols in Refs.~\cite{zhang2024time,cowtan2025parallel}. 

Furthermore, the execution time of a fault-tolerant quantum circuit is heavily dependent on the efficiency of the supporting classical decoding algorithms. 
To decode an EAC block, it would be interesting to consider a modular approach that decodes the code system and extractor system separately, such as the decoder implemented in Ref.~\cite{cross2024improved}.
One could further explore co-designing the code with extractor systems and decoders, with the aim of maximizing encoding rate and measurement efficiency.
For this purpose, belief propagation-based~\cite{panteleev2021degenerate,gong2024toward,wolanski2024ambiguity,hillmann2025localized,muller2025improved} and neural network-based~\cite{liu2019neural,maan2025machine,ott2025decision,blue2025machine} decoders are promising candidates due to their adaptability to different codes and settings.

Throughout this work, we have taken one logical cycle to be $O(d)$ physical syndrome measurement cycles. 
It has been shown that certain families of QLDPC codes can be single-shot decoded~\cite{fawzi2020constant,gu2024single,nguyen2024quantum,bombin2015gauge,Kubica2022single}, which means a logical cycle on these QLDPC memories only uses $O(1)$ syndrome measurement cycles.
Ref.~\cite{hillmann2024single} has shown that lattice surgery can be single-shot in higher dimensional topological codes. 
Can new insights into the theory of code surgery enable broader designs of fast surgery systems? 
The recent works of Refs.~\cite{baspin2025fast,cowtan2025fast} have answered this question in the affirmative by deriving sufficient conditions for surgery operations to be performed in (amortized) constant time overhead. 
An important open direction is to incorporate these ideas into the design of extractor systems.

Recently, Ref.~\cite{yoder2025tour} introduced the \textbf{bicycle architecture}, which can be modeled as an extractor architecture built upon the bivariate bicycle codes of Ref.~\cite{bravyi2024highthreshold}. 
Our works both originated from the work of Ref.~\cite{cross2024improved} and were developed concurrently
with complementary focus: the bicycle architecture is a practical blueprint with extensive optimizations, while the extractor architecture is a theoretical model with vast generality.
Many open directions discussed above were investigated in Ref.~\cite{yoder2025tour} and incorporated into the design of the bicycle architecture.
\newline

To conclude, extractor systems and architectures present a highly flexible, optimizable and scalable proposal for universal fault-tolerant quantum computers based on QLDPC codes.
Over the past two decades, extensive research has established surface code architectures, such as those based on lattice surgery, as the leading proposal to realize fault-tolerant quantum computation in practice.
Could extractor architectures, if optimized and built, compete with surface code architectures
in the large-scale quantum computation regime?
The road ahead presents many exciting problems and opportunities for practical study.

\section*{Acknowledgements}
We thank Patrick Rall for valuable insights and feedback on an earlier version of this paper.
We thank Emily Pritchett, Andrew Cross, Sergey Bravyi and Robert K\"{o}nig for insightful discussions and feedback. 
Z.H. thanks Peter Shor and Anand Natarajan for many inspiring discussions, and thanks Anna Brandenberger, Byron Chin and Xinyu Tan for helpful discussions on the decongestion lemma and the figures.
Z.H. is supported by the MIT Department of Mathematics and the NSF Graduate Research Fellowship
Program under Grant No. 2141064.
D.W.~is currently on leave from The University of Sydney and is employed by PsiQuantum. This work was completed after A.C. began employment at Xanadu.

\newpage

\begin{center}
\textbf{\Large Methods of ``Extractors: QLDPC Architectures for \\
Efficient Pauli-Based Computation''}

\end{center}

\tableofcontents

\section{QLDPC Surgery with Auxiliary Graphs}\label{sec:surgery}

\subsection{Why Logical Measurements?}

QLDPC surgery generally, including the specific extractor system we develop here, enables the fault-tolerant measurement of logical operators in a QLDPC code. For our purposes, we only consider measuring logical Pauli operators, albeit arbitrary ones. Why is this useful?

Unsurprisingly, arbitrary logical Pauli measurements allow reading from a quantum memory. Indeed, we can be selective and measure a single logical qubit of interest, rather than measuring all the logical qubits of a codeblock together as is done via destructive, transversal single-qubit measurement of Calderbank-Shor-Steane (CSS) codes~\cite{calderbank1996good,steane1996multiple}. Relatedly, we can initialize select logical qubits in Pauli eigenstates to be used as ancillas in computation.

More surprising is the fact that arbitrary Pauli measurements along with specially prepared ``magic” resource states can be used for universal quantum computation, in a framework known as Pauli-based computation (PBC)~\cite{bravyi2016trading}. This is similar to the earlier idea of a one-way quantum computer, or measurement-based computation \cite{raussendorf2001one}. The cost of PBC over unitary-based computation is that additional scratch space — ancilla qubits, which we note may be reset and reused, are allocated so that the measurements can avoid collapsing the wavefunction of computational qubits.

So why should we perform fault-tolerant computation using measurements?
The primary reason is that through the works on QLDPC surgery, logical measurement has became a flexible and low-overhead computational primitive.
This primitive can be easily combined with other computational proposals, such as transversal gates, to reach universal computation under the PBC framework. 
In contrast, in the conventional unitary gates framework, gates that cannot be compiled into low-overhead operations need to be implemented by full-block gate teleportation, which incur a heavy overhead.

Moreover, performing logical Pauli measurements is, in a sense, a very natural operation on quantum codes. Throughout the lifetime of a quantum memory, we are measuring stabilizer checks to collect syndrome information for the purpose of correcting errors. 
QLDPC surgery postulates simply changing the pattern of those check measurements, and therefore the quantum code, over a period of time and involving a slightly larger set of qubits so as to extract not just the syndrome but also the eigenvalue of a logical Pauli operator~\cite{hastings2021weight,cohen2022low}. Successively altering the measurement pattern measures different logical operators over the course of a computation. 

For the remainder of Section~\ref{sec:surgery}, we review prior and recent works on QLDPC surgery and lay the foundation for our constructions.

\subsection{Prior Works}\label{sec:prior_works}

The techniques in previous works on QLDPC surgery can be described in a unifying framework, which we summarize here. 
Let $L$ be the support of a logical Pauli operator $\ML$, which could be a product of Pauli operators on multiple code blocks.
We construct a \textbf{measurement hypergraph} $H = (V, \ME)$, with vertices $V$ and hyperedges $\ME$, and an injective function $f: L\rightarrow V$, which we call the \textbf{port function}. 
We place a qubit on every hyperedge $h\in \ME$ and design two types of stabilizer checks: \textbf{vertex checks} and \textbf{cycle checks}. 
As their names suggest, each vertex $v\in V$ is associated to a vertex check $A_v$, which is supported on all hyperedges $h\ni v$, and each cycle\footnote{A cycle in a hypergraph is a collection of edges that contains every vertex an even number of times.} $C$ is associated to a cycle check $B_C$, which is supported on all hyperedges $h\in C$. 
We then connect this ancilla system to the code $\qcode$ through the port function:
for every qubit $q\in L$, we further extend the vertex check $A_{f(q)}$ to act on $q$. 
This completes the description of the merged code $\bar{\qcode}$, in which the operator $\ML$ become a product of constant-weight stabilizer checks. 
See Definition~\ref{def:graph_and_code} for full details.

In Ref.~\cite{cohen2022low}, Cohen, Kim, Bartlett and Brown first considered the case where $\qcode$ is CSS and $\ML$ is irreducible.\footnote{A logical operator $\ML$ supported on a set of qubits $L$ is irreducible if the restriction of $\ML$ to any proper subset of qubits in $L$ is not a logical operator or stabilizer.\label{footnote:irreducible}}
They used the \textbf{induced Tanner graph} $T = (V, E)$ of $\ML$, which is a hypergraph with $V$ isomorphic to $L$ and, in the case where $\ML$ is an $X$ operator, all $Z$ checks overlapping with $\ML$ as hyperedges (there is an analogous construction when $\ML$ is a $Z$ operator). 
The authors took a copy $T_1 = (V_1, E_1)$ of $T$ and \textbf{thickened} $T_1$ with a line graph $J_d$ of length $d$ (see Definition~\ref{def:thicken}). 
They then used the 1-to-1 port function $f:V\rightarrow V_1$ and measured all vertex checks $A_v$ as $X$-checks, and a selected set of cycle checks $B_C$ as $Z$-checks. 
From an equivalent perspective (see Remark~\ref{rmk:perspectives}), the checks defined in Ref.~\cite{cohen2022low} correspond to the checks of a hypergraph product code defined from $T_1$ and $J_d$, extended by the port function $f$ onto $\qcode$. 
To measure a product of logical operators on multiple logical qubits, the authors connected (and in some cases merged) individual measurement hypergraphs by adding more vertices and cycles.\footnote{The added checks may be non-CSS.} The code switching protocol between the code $\qcode$ and the merged code $\bar{\qcode}$ is then fault-tolerant. We henceforth refer to the construction in Ref.~\cite{cohen2022low} as the \textit{CKBB scheme}.

Since the induced Tanner graph is thickened by $J_d$, the measurement graph in the CKBB scheme has a daunting size of $O(|L|d)$.
Consequently, measuring a weight $d$ logical operator uses $O(d^2)$ ancilla qubits.  
Most later works on QLDPC surgery, including the recent advances, are motivated by reducing this space overhead.
Refs.~\cite{cowtan2024css,cowtan2024ssip} considered using a shorter line graph for thickening in the CKBB scheme, and constructing a hypergraph directly between codeblocks.
While theoretically this approach does not guarantee fault-tolerance, Ref.~\cite{cowtan2024ssip} demonstrated numerically that on various small-to-medium QLDPC codes, the merged code distance can be preserved with a smaller CKBB measurement graph.
We refer readers to Ref.~\cite[App.~D]{cowtan2024ssip} for a list of improvements in space overheads.

Shortly after Ref.~\cite{cowtan2024ssip}, the independent work Ref.~\cite{cross2024improved} presented the \textit{gauge-fixed surgery scheme}.
Ref.~\cite{cross2024improved} observed that if $T$ is expanding and we measure \textit{all} cycle checks $B_C$, then we can thicken $T_1$ with a much shorter line graph and maintain fault-tolerance of the overall protocol. 
These observations led to a qualitative improvement in the space overhead of QLDPC surgery, which in the case of the $[[144, 12, 12]]$ bivariate bicycle code~\cite{bravyi2024highthreshold} reduced the size of the ancilla system from $1380$ qubits to $103$ qubits~\cite{cross2024improved}. 
A caveat, however, is that some cycle checks could have large weight, and the scheme therefore lacks guarantee of being LDPC.
Similar to the CKBB scheme, the gauge-fixed surgery scheme assumed that $\qcode$ is CSS and $\ML$ is irreducible. 
To measure product of logical Paulis, the authors observed that the methods from Ref.~\cite{cohen2022low} are no longer fault-tolerant when the path graph used in thickening has length less than $d$. 
Alternatively, they proposed to add a \textbf{bridge} system, which under this framework corresponds to a set of $d$ edges, to connect individual measurement hypergraphs. 
This addition enables us to perform product measurements on logical qubits in the same code block or different code blocks.
Moreover, the measured logical qubits may belong to different QLDPC code families, in which case the bridge system serves as an adapter of codes.
One notable caveat, however, is that adding a bridge of edges creates new cycles in the measurement hypergraphs, which are not guaranteed to be low weight. 
This lack of an LDPC guarantee was later resolved in Ref.~\cite{swaroop2024universal} using a novel SkipTree algorithm.

In the \textit{gauging measurement scheme}~\cite{williamson2024low} and independently in the \textit{homological measurement scheme}~\cite{ide2024fault}, the authors proposed to replace the induced Tanner graph of the measured operator by a customized expander graph\footnote{In general this customized graph can be a hypergraph, but a simple graph is easier to work with.} (with a customized port function). 
This is a qualitative change for two reasons. 
First, we no longer need to rely on $T$ having expansion and can instead inject expansion via the measurement graph. 
Moreover, some assumptions made in previous works, namely that $\qcode$ is CSS and $\ML$ is irreducible, can now be relaxed. 
As a result, product measurements can be handled in the same way as single qubit logical measurements.\footnote{While this procedure applies directly to logicals on disjoint code blocks, the bridge system is still useful for its low overhead and modularity.} 
The works then measured all cycle checks assisted by the technique of \textbf{cellulation} (Definition~\ref{def:cellulation}). 
The scheme in Ref.~\cite{williamson2024low} further applied the techniques of \textbf{decongestion} (Lemma~\ref{lem:decongestion}) after thickening the customized expander graph by a path graph of length $O((\log |L|)^3)$, to guarantee that the resulting merged code is LDPC, whereas the schemes in Refs.~\cite{cross2024improved, ide2024fault} lack this guarantee in worst case. 
Consequently, the space overhead of measuring an operator $\ML$ is reduced to $O(|L|(\log |L|)^3)$ with an LDPC guarantee in the worst case.

Ref.~\cite{swaroop2024universal}, building upon the scheme of Ref.~\cite{williamson2024low}, 
showed that a bridge/adapter system can always be constructed between measurement graphs\footnote{If the measurement graph is a hypergraph, the techniques of Ref.~\cite{swaroop2024universal} no longer work and we once again lose an LDPC guarantee. Nonetheless, the measurement graphs constructed by the main procedure in Ref.~\cite{williamson2024low} are always simple graphs, so the bridge/adapter system can always be chosen to be LDPC.} so that the newly created cycles admit a low weight basis. 
Consequently, the \textbf{adapter} construction becomes truly universal in the sense that it can connect two (or more) code blocks from arbitrary code families together into one LDPC and fault-tolerant architecture. 
Such a diversified architecture could take advantage of different codes for different aspects of computation.
We present such architectures in Section~\ref{sec:architectures}.
Besides the bridge/adapter system, Ref.~\cite{swaroop2024universal} further improved many ideas from Ref.~\cite{cross2024improved} and Ref.~\cite{williamson2024low}, including relative expansion (Definition~\ref{def:relative_expansion}), port function and graph desiderata (Theorem~\ref{thm:graph_desiderata}). 

The work of Ref.~\cite{zhang2024time} studied the problem of measuring a collection of $Z$ (or $X$) logical Pauli product operators simultaneously. 
They proposed several techniques, including branching and devised sticking, which when combined with the CKBB scheme enables simultaneous measurements at various overheads.
They further extend their scheme to measure arbitrary commuting subgroup of Pauli operators using the technique of twist-free lattice surgery~\cite{Chamberland2022twist}, at the cost of potentially expensive preparation of $\ket{Y}$ states. 
Branching, in our measurement hypergraph framework, is equivalent to thickening an induced Tanner graph with an open segment, which is the graph with one vertex and one edge attached to the vertex. 
Attaching such a \textit{branching sticker} to a logical operator $\ML$ creates new representatives of $\ML$ on the ancilla qubits, which makes it a useful primitive for surgery.

The work of Ref.~\cite{cowtan2025parallel} improved the technique of branching,
and combined it with gauging measurements instead of the CKBB scheme. This led to a significant reduction in space overhead for simultaneous measurements, and a qualitative improvement in capacity -- the scheme in Ref.~\cite{cowtan2025parallel} can measure Pauli products with $Y$ terms in parallel.
As a result, when applying twist-free lattice surgery to measure arbitrary commuting subgroup of Pauli operators, the $\ket{Y}$ states needed are now much cheaper to prepare. 

We emphasize that the papers discussed have many additional contributions not captured by our summary above.
Moreover, parallel to the developments in code surgery, another technique to perform logical measurements called 
homomorphic measurements~\cite{huang2023homomorphic} has been developed. This measures select logical operators in a QLDPC code by creating a logical ancilla, encoded generally in a different but related code, which is then coupled to the original code via transversal gates and measured out. Importantly, it is not always known how to create a suitable logical ancilla state, though it can be done on topological codes~\cite{huang2023homomorphic} and notably for some measurements on homological product codes~\cite{xu2024fast}. QLDPC surgery likely offers another avenue for preparing the required ancilla states.

\begin{remark}[Equivalent perspectives on QLDPC surgery]\label{rmk:perspectives}
    The original, and most common, view on QLDPC surgery is through Tanner graphs, whereby data qubits and stabilizer checks on a code are assigned vertices in a bipartite graph. In this picture, surgery operations can be described by adding vertices and edges to the graph~\cite{cohen2022low,cross2024improved}.
    Previous works have studied surgery on QLDPC codes which are CSS through the lens of homology, using the bijection between qubit CSS codes and chain complexes over $\mathbb{F}_2$. This was the view taken in e.g. Refs.~\cite{cowtan2024css, cowtan2024ssip, ide2024fault}. 
    While the two perspectives are equivalent in the CSS surgery cases, a helpful property of chain complexes is that they come with well-defined chain maps -- maps between codes -- which allow for certain convenient proofs \cite{cowtan2024css,ide2024fault} concerning, for example, how logical operators relate between the original code and the deformed code. On the other hand, Tanner graphs can be visualized easily, and provide simple descriptions for the non-CSS surgery cases.
\end{remark}

\subsection{Surgery Toolkit: Logical Measurements}\label{sec:toolkit_measurements}

Following these recent developments, in the rest of this section we package a collection of definitions and results into a toolkit for the design and analysis of QLDPC surgery schemes. 
We emphasize that this toolkit
in no way subsumes the many perspectives and techniques developed in prior works. 
Nonetheless, as discussed in Section~\ref{sec:prior_works}, many key results can be described using these definitions.
In particular, this toolkit establishes the foundation of the analysis in Refs.~\cite{cross2024improved, williamson2024low, swaroop2024universal} and our main results in this paper.
In future works, we plan to build upon this toolkit and expand it with additional existing and new techniques.

Let $\qcode$ be a QLDPC code with parameters $[[n, k, d]]$ and stabilizer checks $\MS$.\footnote{Here $\MS$ denotes the set of stabilizer checks to be measured on the code, not the entire stabilizer group.}
Let $Q$ denote the set of qubits of $\qcode$.
Let $\ML$ be a Pauli logical operator of $\qcode$ with support $L$.
Here, we make no assumption on $\ML$: it can be a product of logical Pauli $X, Y, Z$ operators on any representatives of any logical qubits of $\qcode$.
We use $\ML_q\in \{I, X, Y, Z\}$ to denote the action of $\ML$ on qubit $q\in Q$.
For a qubit $q$, let $Z(q)$ denote the Pauli operator that acts on $q$ by $Z$ and acts on all other qubits by identity.
We extend this notation to $X, Y, I$ and to sets of qubits.
In our notation, $\ML = \prod_{q\in Q} \ML_q(q)$. 

\begin{definition}[Measurement Graphs and Codes]\label{def:graph_and_code}
    Consider a graph $G = (V, E)$ and an injective function $f: L\rightarrow V$. 
    We call $G$ the \textbf{measurement graph} and $f$ the \textbf{port function}.
    We say that $P = \im(f)$ is the \textbf{port}.
    Create an ancilla qubit for every $e\in E$. For notational convenience, we use $e$ and $E$ to denote both the edge(s) and the ancilla qubit(s). 
    We define a stabilizer code $\bar{\qcode}$ supported on $Q\cup E$ with the following stabilizers $\bar{\MS}$. 
    \begin{enumerate}[itemsep = 0pt]
        \item For vertices $v\in V$, 
        \begin{enumerate}[itemsep = 0pt]
            \item if $v\notin P$, add the stabilizer $A_{v} = \prod_{e\ni v}Z(e)$ to $\bar{\MS}$. \label{stabilizers:non-port-vtxs}

            \item If $v = f(q)$ for $q\in Q$, add the stabilizer $A_{v} = \ML_{q}(q)\prod_{e\ni v}Z(e)$ to $\bar{\MS}$. \label{stabilizers:port-vtxs}
        \end{enumerate}
        We refer to these checks as the \textbf{vertex checks}.

        \item Let $\MR$ be a cycle basis (Definition~\ref{def:cycle_basis}) of $G$. For every cycle $C\in \MR$, add stabilizer $B_C = \prod_{e\in C}X(e)$ to $\bar{\MS}$. We refer to these checks as the \textbf{cycle checks}. \label{stabilizers:cycles}

        \item For every check $S\in \MS$ of $\MQ$, let $K(S, \ML)$ denote the set of qubits $q\in Q$ such that $S_q$ and $\ML_q$ anti-commutes. Note that $|K(S, \ML)|$ must be even. 
        \begin{enumerate}[itemsep = 0pt]
            \item If $K(S, \ML) = \varnothing$, add $S$ to $\bar{\MS}$. \label{stabilizers:code_checks}
            
            \item Otherwise, let $\mu(S, \ML)$ be a path matching (Definition~\ref{def:path_matching}) of $f(K(S, \ML))$\footnote{For a set $K$ of qubits, we define $f(K) = \{f(q): q\in K\}$.} in $G$. Add the stabilizer $\bar{S} = S\prod_{e\in \mu(S, \ML)}X(e)$ to $\bar{\MS}$. \label{stab:modified_code_checks}
            
        \end{enumerate}
    \end{enumerate}
    For clarity, we sometimes denote $\bar{\qcode}$ as $\qcode(\ML, G, f)$ and $\bar{\MS}$ as $\MS(\ML, G, f)$. 
\end{definition}

\begin{remark}
    We note that equivalently, we could define the vertex and cycle checks to act on edge qubits by $X$ and $Z$, respectively. The stabilizers defined in Step~\ref{stab:modified_code_checks} should then act on edge qubits by $Z$. All following results hold with respect to either basis choice. 
\end{remark}

\begin{definition}[Cycle Basis and Congestion]\label{def:cycle_basis}
    For a (hyper)graph $G = (V, E)$, consider its incidence matrix $M_G\in \FF_2^{|V|\times |E|}$ defined by 
    \begin{equation}
        M_G[v, e] = \begin{cases}
            1 &\text{ if $v\in e$,} \\
            0 &\text{ otherwise}. 
        \end{cases}
    \end{equation}
    Then the kernel of $M_G$ is precisely the space of cycles in $G$.  
    A basis $\MR$ of $\ker(M_G)$ is called a \textbf{cycle basis} of $G$. 
    For such a basis, for every (hyper)edge $e\in E$, let $\rho_\MR(e)$ denote the number of times $e$ is used by cycles in $\MR$. 
    Let $\rho = \max_{e\in E} \rho_\MR(e)$, we say that $\MR$ has \textbf{congestion} $\rho$, or is a $\rho$-basis~\cite{reich2014cycle}.
\end{definition}

\begin{definition}[Path Matching]\label{def:path_matching}
    For a (hyper)graph $G = (V, E)$ and a set of vertices $K\subseteq V$, a path matching $\mu$ of $K$ is a collection of (hyper)edges in $G$ that visits every vertex in $K$ an odd number of times, and every vertex in $V\setminus K$ an even number of times.
\end{definition}

\begin{remark}
    The present formulation of Definition~\ref{def:graph_and_code} mostly follows the formulation of gauging measurement in Refs.~\cite{williamson2024low} and~\cite{swaroop2024universal}. 
    The independent work Ref.~\cite{ide2024fault} formulated a similar scheme called homological measurement. 
    As the ideas are developed progressively through previous works, we simply refer to them as measurement graphs and codes. 
\end{remark}

\begin{figure}
    \centering
    \includegraphics[width = 0.52\textwidth]{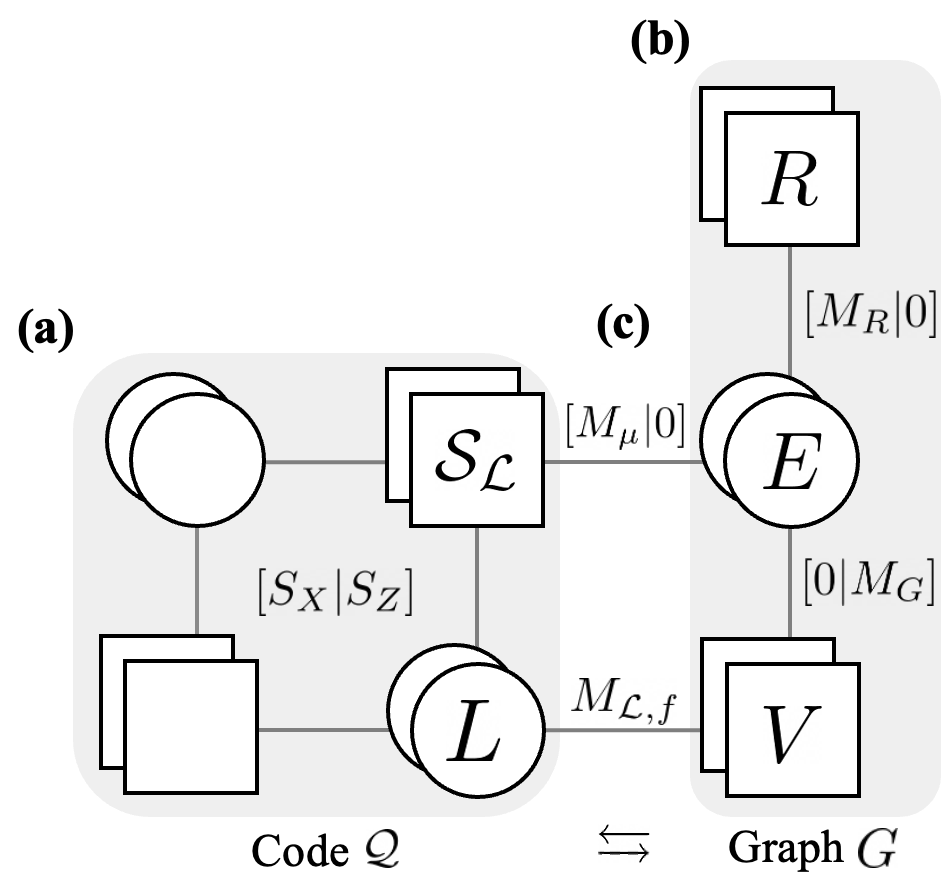}
    \caption{Logical measurement of an operator $\ML$ using a measurement graph $G$, depicted with scalable Tanner graphs. Here, groups of circles denote qubits, and groups of squares denote checks. Lines between checks and qubits are labelled by sympletic matrices, denoting the Pauli actions the checks have on the qubits.
    \textbf{(a)} Tanner graph of the code $\qcode$. The qubits on the right labelled $L$ are qubits in support of $\ML$, and checks on the right labelled $\MS_\ML$ are checks where $K(\MS, \ML)\ne\varnothing$ (see Definition~\ref{def:graph_and_code}, checks~\ref{stab:modified_code_checks}). Unlabelled qubits on the left are the remaining qubits in $Q\setminus L$, unlabelled checks on the left are remaining checks with $K(\MS, \ML) = \varnothing$. Checks act on qubits as specified by the sympletic stabilizer matrix $[S_X\vert S_Z]$ of $\qcode$. 
    \textbf{(b)} Ancilla system specified by the measurement graph $G$. Every edge in $G$ is an ancilla qubit. Vertex checks $V$ act on edge qubits by $Z$ with incidence matrix $M_G$, cycle checks from basis $R$ act on edge qubits by $X$ with incidence matrix $M_R$.
    \textbf{(c)} The code and graph systems are coupled by check deformation. The vertex checks $V$ act on qubits in $L$ as specified by the port function $f$ and the operator $\ML$. The code checks $\MS_\ML$ act on edge qubits that form path matchings by $X$ as specified by Definition~\ref{def:graph_and_code}, checks~\ref{stab:modified_code_checks}.
    }
    \label{fig:gauging}
\end{figure}

The motivation behind defining the code $\qcode(\ML, G, f)$ is simple. 
As proved in previous works (and as one can easily verify), the stabilizers defined in Definition~\ref{def:graph_and_code} commute. 
Moreover, the logical operator $\ML$ which we wish to measure is now a product of stabilizers in $\qcode(\ML, G, f)$. 
Specifically, we have
\begin{equation}\label{eq:measured_operator}
    \ML = \prod_{v\in V} A_v. 
\end{equation}
Therefore, by performing a code-switching measurement protocol between $\qcode$ and $\qcode(\ML, G, f)$ (Definition~\ref{def:measurement_protocol}), we can fault-tolerantly perform the logical measurement of $\ML$. 
To facilitate such a protocol, the measurement graph $G$ needs to satisfy the following graph desiderata.

\begin{theorem}[Graph Desiderata~\cite{williamson2024low,swaroop2024universal}]\label{thm:graph_desiderata}
    For any logical operator $\ML$ with support $L$, any graph $G = (V, E)$ and port function $f: L\rightarrow V$, the code $\bar{\qcode} = \qcode(\ML, G, f)$ is well-defined and $\ML$ is a product of stabilizers in $\bar{\qcode}$. 
    Moreover:
    \begin{enumerate}[itemsep = 0pt]
        \item If $G$ is connected, then $\bar{\qcode}$ encodes the remaining $k-1$ logical qubits of $\qcode$ after measurement of $\ML$. 
        More precisely, every logical operator $\ML'\ne \ML$ of $\qcode$ which commutes with $\ML$ has an independent equivalence class in $\bar{\qcode}$.
        \label{des:connected}

        \item \label{des:LDPC} 
        To ensure $\bar{\qcode}$ is LDPC, it is necessary and sufficient that
        \begin{enumerate}[itemsep = 0pt]
            \item The maximum degree of $G$ is $O(1)$; \label{des:degree}

            \item There is a cycle basis $\MR$ of $G$ such that $\MR$ has congestion $O(1)$, and every cycle in $\MR$ has length $O(1)$. \label{des:cycle_basis}

            \item Consider the collection of path matchings $\mu(S, \ML)$ for all checks $S\in \MS$. Every $\mu(S, \ML)$ has $O(1)$ edges and every edge is in $O(1)$ path matchings. \label{des:path_matching}
        \end{enumerate}

        \item To ensure $\bar{\qcode}$ has distance at least $d$, it is sufficient for the relative Cheeger constant (Definition~\ref{def:relative_expansion}) $\beta_d(G, f(L))$ to be at least $1$. \label{des:relative_expansion}
    \end{enumerate}
\end{theorem}

\begin{definition}[Cheeger Constant and Relative Expansion]\label{def:relative_expansion}
For a graph $G = (V, E)$, for a set of vertices $U\subseteq V$, the \textbf{edge boundary} of $U$, which we denote $\delta_G U$, is defined as the number of edges with exactly one endpoint in $U$.
When the graph $G$ is clear from context, we simply write $\delta U$.
The \textbf{Cheeger constant} $\beta(G)$ is defined as the largest real number such that for all $U\subseteq V$, 
\begin{equation}
    |\delta U| \ge \beta(G)\cdot \min(|U|, |V\setminus U|). 
\end{equation}
Furthermore, for a subset of vertices $P\subseteq V$ and an integer $t$, we define the \textbf{relative Cheeger constant} $\beta_t(G, P)$ to be the largest real number such that for all $U\subset V$, we have
\begin{equation}
    |\delta U|\ge \beta_t(G, P)\cdot \min(t, |U\cap P|, |P\setminus U|).
\end{equation}
From the definitions, we see that $\beta(G) = \beta_{|V|}(G, V)$. 
\end{definition}
\begin{remark}\label{rmk:practical_conditions}
    The desiderata are sufficient conditions for the proof of Theorem~\ref{thm:graph_desiderata} and later Theorem~\ref{thm:fault_distance}. In practice, most of these conditions can be relaxed, as we discuss in Section~\ref{sec:toolkit_practical}.
\end{remark}

The notion of relative expansion in the above form was introduced in Ref.~\cite{swaroop2024universal}.
For a proof of Theorem~\ref{thm:graph_desiderata}, we refer readers to Section~2.3 and Appendix~A of Ref.~\cite{swaroop2024universal}, and note that similar lemmas were proved in Ref.~\cite{cross2024improved} and Ref.~\cite{williamson2024low}.

Given a measurement graph which satisfy the desiderata, the following protocol performs a logical measurement of $\ML$ when it is implemented noiselessly.
\begin{definition}[Measurement Protocol~\cite{williamson2024low}]\label{def:measurement_protocol}
Given a state $\ket{\Psi}$ in the code space of $\qcode$, a logical operator $\ML$, a measurement graph $G$ and a port function $f$, the following procedure outputs $\sigma = \pm 1$ as the result of measuring $\ML$ and the resulting code state $\frac{1}{2}(\mathbbm{1} + \sigma\ML) \ket{\Psi}$.
\begin{enumerate}[itemsep = 0pt]
    \item Initialization: Prepare all edge qubits in $\ket{0}$. \label{stage:init}
    \item Merge: Measure the stabilizers $\MS(\ML, G, f)$. \label{stage:merge}
    For vertex checks $A_v$, record its measurement result as $\epsilon_v = \pm 1$. Output $\sigma = \prod_{v\in V}\epsilon_v$.
    \item Split: Measure all edge qubits in $Z$ basis. For each edge $e$, record the measurement result as $\omega_e$. \label{stage:split}
    \item Correct: Fix an arbitrary vertex $v_0 \in V$. For every qubit $q\in L$, let $\gamma$ be an arbitrary path of edges from $v_0$ to $f(q)$. If $\prod_{e\in \gamma} \omega_e = -1$, apply single-qubit correction $X(q)$. \label{stage:correction}
\end{enumerate}

\end{definition}
We briefly remark on the last correction stage of the above protocol. The edge qubit measurements in the splitting stage anti-commutes with the stabilizer checks in $\MS(\ML, G, f)$ defined in Step~\ref{stab:modified_code_checks} of Definition~\ref{def:graph_and_code}.
Therefore, the results $\omega_e$ are intrinsically non-deterministic, and the random collapse of the edge qubits into the $Z$ basis induces $X$ errors (or byproduct operators) on $Q$. 
These errors are corrected in stage~\ref{stage:correction}.

We add cycles of error correction to this measurement protocol to make it fault-tolerant.
In this work, we measure fault-tolerance with the notion of \textbf{phenomenological fault distance}, which is also called space-time fault distance.
For an error-corrected protocol, this is defined as the minimum number of qubit errors and measurement errors needed to cause an undetected logical error (which includes getting an incorrect logical measurement result). 
\begin{theorem}[Fault-Tolerance~\cite{williamson2024low}]\label{thm:fault_distance}
    Suppose the measurement graph $G$ and port function $f$ satisfy the desiderata of Theorem~\ref{thm:graph_desiderata}. 
    To implement the measurement protocol (Definition~\ref{def:measurement_protocol}) fault-tolerantly, we perform $d$ rounds of syndrome measurement cycles for $\qcode$ (followed by decoding and correction) before stage~\ref{stage:init} (Initialization) and after stage~\ref{stage:correction} (Correction). 
    We also measure the stabilizers $\MS(\ML, G, f)$ for $d$ rounds during stage~\ref{stage:merge} (Merge). After decoding and correction, we output the measurement result $\sigma$. This fault-tolerant implementation of the measurement protocol has space-time fault distance $d$. 
\end{theorem}
The above theorem is stated and proved as Theorem~1 and~2 in Ref.~\cite{williamson2024low}.
Beyond distance, we also would like decoders for $\qcode$ and $\qcode(\ML, G, f)$ which are capable of correcting clusters of space-time errors of weight $O(d)$ or stochastic errors.
An example is the modular decoder proposed and implemented in Ref.~\cite{cross2024improved}, which decodes $\qcode(\ML, G, f)$ up to fault-distance $d/2$ assuming that we are given a decoder which decodes $\qcode$ up to fault-distance $d/2$.
It would be interesting to investigate whether the modular decoder can be adapted to the more general setting of non-CSS codes.

\begin{remark}\label{rmk:multiple_code_blocks}
    While we have phrased the results in this section as measuring a logical operator $\ML$ on a single code block of $\qcode$, all results hold if $\ML$ is supported on multiple code blocks of different code families.
    Pedantically, let $\qcode_1, \cdots, \qcode_B$ be quantum codes with qubits $Q_1, \cdots, Q_B$ and stabilizers $\MS_1, \cdots, \MS_B$. 
    We denote $\qcode_1\cup\cdots\cup\qcode_B$ as the joint code on qubits $Q_1\cup \cdots \cup Q_B$ with stabilizers $S_1\cup \cdots \cup S_B$, where every stabilizer $S\in S_i$ acts on $Q_i$ as before and acts on the remaining qubits by identity. 
    All previous analysis hold if $\ML$ is a logical Pauli operator supported on such a joint code.
\end{remark}

\subsection{Surgery Toolkit: Building Measurement Graphs}\label{sec:toolkit_graphs}

We now discuss how to construct measurement graphs satisfying the desiderata (Theorem~\ref{thm:graph_desiderata}).
Desideratum~\ref{des:connected} is straightforward to satisfy. 
To satisfy Desideratum~\ref{des:relative_expansion}, we state the following lemma on relative expansion. 

\begin{lemma}[Restriction Lemma]\label{lem:restriction_relative_expansion}
    Fix a graph $G = (V, E)$. Let $P, P'$ be subsets of $V$ with $P\subseteq P'$, and $t, t'$ be integers such where $t \le t'$. 
    Then $\beta_t(G, P)\ge \beta_{t'}(G, P')$.
\end{lemma}
\begin{proof}
    The proof follows directly from the definition of relative expansion. Note that for $P\subseteq P'$, we have for all $U$, $U\cap P\subseteq U\cap P'$ and $P\setminus U\subseteq P'\setminus U$. 
    Since $t \le t'$, we have 
    \begin{equation}
        \min(t, |U\cap P|, |P\setminus U|) \le \min(t', |U\cap P'|, |P'\setminus U|),
    \end{equation}
    which implies that $\beta_t(G, P)\ge \beta_{t'}(G, P')$.
\end{proof}

A direct corollary of the above lemma is that for all $P\subset V, t\le |V|$, we have $\beta_t(G, P)\ge \beta(G)$.

We now discuss how to satisfy Desideratum~\ref{des:LDPC} while preserving or improving relative expansion.
For Desideratum~\ref{des:cycle_basis}, we cite the following Decongestion Lemma, which is Lemma~A.0.2 in Ref.~\cite{freedman2021building}.
The precise bounds can be obtained by an inspection of the relevant proof in Ref.~\cite{freedman2021building}.
\begin{lemma}[Decongestion Lemma~\cite{freedman2021building}]\label{lem:decongestion}
    Fix any simple graph $G = (V, E)$, $G$ has a cycle basis $\MR$ with congestion $\rho < \log_2(|V|)
    \ln(2|E|)$. 
    Moreover, say that two cycles overlap if they share edges. There is an ordering of the basis, $\MR = \{C_1, \cdots, C_{|E|-|V|+1}\}$, such that every cycle $C_i$ overlaps with at most $\log_2(|V|)\cdot \rho$ cycles later in the ordering.
    Such an ordered basis can be found with an efficient randomized algorithm.
\end{lemma}
We include a proof of the following corollary in Appendix~\ref{apdx:proofs}.
\begin{restatable}{corollary}{cordecongestion}\label{cor:decongestion}
    We can efficiently compute a partition of $\MR = \bigcup_{i=1}^t \MR_i$ such that each $\MR_i$ contains non-overlapping cycles and $t\le \log_2(|V|)\cdot \rho+1$.
\end{restatable}
Note that the result of the above lemma alone is insufficient to satisfy Desideratum~\ref{des:cycle_basis}, as the congestion is not constant and the cycles in $\MR$ may have arbitrary length.
To further decongest the cycles, we apply the technique of thickening.

\begin{definition}[Thickening]\label{def:thicken}
    Consider two graphs $G_1 = (V_1, E_1)$ and $G_2 = (V_2, E_2)$. 
    The Cartesian product of $G_1$ and $G_2$ is the graph $G = G_1\square G_2 = (V_1\times V_2, E)$ where 
    \begin{equation}
        E = \{((u_1,u_2), (v_1,v_2)): (u_1,v_1)\in E_1\text{ or }(u_2,v_2)\in E_2\}.
    \end{equation}
    A line graph of length $\ell$ is a graph $J_\ell$ with $\ell$ vertices $v_1, \cdots, v_{\ell}$ and $\ell-1$ edges $(v_1,v_2), \cdots, (v_{\ell-1}, v_{\ell})$. 
    We say that $G\square J_\ell$ is $G$ \textbf{thickened} $\ell$ times (See Figure~\ref{fig:thicken_cellulate}b for an example).
    We refer to the $\ell$ copies of $G$ as levels, and denote them $G\times \{r\} = (V\times \{r\}, E\times \{r\})$ for $1\le r\le \ell$.
\end{definition}

The next fact explains the motivation behind thickening a graph.
We include a proof in Appendix~\ref{apdx:proofs}.
\begin{restatable}{fact}{thickenbasis}\label{fact:thickened_cycles}
    Fix a graph $G = (V, E)$.
    In the thickened graph $G\square J_\ell$, consider the following set of length-4 cycles (labelled by their endpoints).
    \begin{equation}
        \MT = \{ (v\times \{r\},  v\times \{r+1\}, u\times \{r+1\}, u\times \{r\}): (v,u)\in E, 1\le r\le \ell-1\}. 
    \end{equation}
    Let $\MR = \{C_1, \cdots, C_{|E|-|V|+1}\}$ be a cycle basis of $G$.
    For every cycle $C_i$, choose an arbitrary level $1\le r_i\le \ell$. 
    Then the set $\MT\cup \{C_i\times \{r_i\}: C_i\in \MR\}$ is a cycle basis of $G\square J_\ell$. See Figure~\ref{fig:thicken_cellulate}c for an example.
\end{restatable}
This fact enables us to measure the cycles in $\MR$ on any level of the thickened graph, which is crucial for satisfying Desideratum~\ref{des:cycle_basis}.
Since the cycles in $\MR$ could have arbitrary length, we add edges to break them into $\FF_2$-sums of constant-weight cycles.

\begin{definition}[Cellulation]\label{def:cellulation}
    Given a simple cycle\footnote{A cycle is simple if it is a (connected) path which visits every vertex exactly 0 or 2 times.} $C$ in a graph $G = (V, E)$, suppose $C$ traverses vertices $1, \cdots, w$ in order.
    We can \textbf{cellulate} $C$ by adding edges $(1, w-1), (w-1, 2), (2, w-2), (w-2, 3), \cdots $, so that $C$ is decomposed into $w-2$ many triangles. 
    Observe that on every vertex we added at most $2$ edges, and every edge of $C$ is used in exactly one triangle. 
    Every added edge gets used at most twice in the triangles. See Figure~\ref{fig:thicken_cellulate}d for an example.
\end{definition}

We can now put the techniques together to construct a measurement graph satisfying the desiderata of Theorem~\ref{thm:graph_desiderata}.
\begin{lemma}[Measurement Graph Construction~\cite{williamson2024low}]\label{lem:measurement_graph_construction}
    For a logical operator $\ML$ with support $L$, we use the following procedure to construct a graph $G$ and a port function $f$ which satisfy the desiderata of Theorem~\ref{thm:graph_desiderata}.
    \begin{enumerate}[itemsep = 0pt]
        \item Let $V_1$ be a set of vertices of size $|L|$, and construct a bijection $f$ between $L$ and $V_1$.
        \item \label{construction:base} Construct a base graph $G_1 = (V_1, E_1)$ as follows: 
        \begin{enumerate}[itemsep = 0pt]
            \item For every stabilizer $S\in \MS$, recall that $K(S, \ML)$ is the set of qubits $q\in Q$ such that $S_q$ and $\ML_q$ anti-commutes.
            Add a perfect matching $\mu(S, \ML)$ of $f(K(S, \ML))$, which is a set of $|K(S, \ML)|/2$ edges, to $G_1$.
            \label{construction:matching}
            \item Construct a constant degree graph $D$ on $|L|$ vertices with Cheeger constant $\beta_D\ge \beta$ for some constant $\beta$ of our choice. Add the edges of $D$ to $G_1$.
            \label{construction:expand}
        \end{enumerate}
        \item \label{construction:decongestion} Apply the Decongestion Lemma~\ref{lem:decongestion} and Corollary~\ref{cor:decongestion} to obtain a cycle basis of $\MR$ with congestion $\rho$ of $G_1$, 
        and a partition $\MR = \bigcup_{i=1}^t \MR_i$ such that each $\MR_i$ contains non-overlapping cycles and $t\le O((\log |L|)^3)$.
        \item \label{construction:thicken} Thicken $G_1$ by $\ell = \max(t, 1/\beta)$ times to obtain $G = G_1\square J_\ell$, denote $G = (V, E)$.
        \item \label{construction:cellulation} On every level $G\times \{r\}$, cellulate every cycle in $\MR_r$ to obtain a collection of triangles which generate the cycles $\MR_r\times \{r\}$.
    \end{enumerate}
    Since $\qcode$ is a LDPC code, we assume every stabilizer check $S$ has weight at most $\omega$ and every qubit in $Q$ is checked by at most $\Delta$ stabilizers. 
    Suppose the expander graph we used in Step~\ref{construction:expand} has maximum degree $\delta$.
    Then the constructed graph $G$ has maximum degree at most $2(\Delta+\delta+1)$ and total edges at most $\ell(\Delta+\delta+1)|L|\le O(|L|(\log |L|)^3)$. 
    The cycle basis we measure has congestion $2$ and maximum length $4$.\footnote{Note that the maximum stabilizer weight of $\qcode$ does not impact these upper bounds.}
\end{lemma}
\begin{figure}
    \centering
    \includegraphics[scale=0.4]{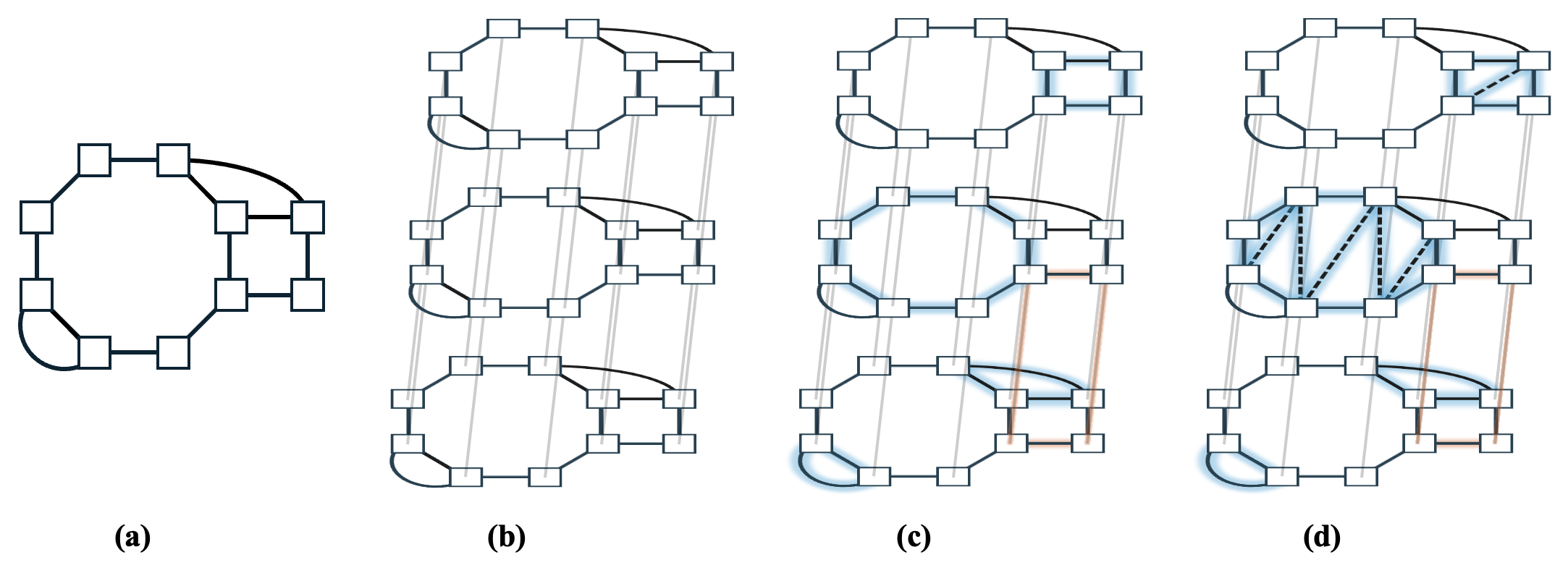}
    \caption{Depicting of thickening, decongestion, and cellulation. 
    \textbf{(a)} A generic graph $G_1$. 
    \textbf{(b)} $G_1$ thickened by a line graph $J_3$, $G = G_1\square J_3$. 
    \textbf{(c)} A cycle basis in $G$, where the blue cycles are a cycle basis of $G_1$ spread into distinct levels and therefore do not overlap, and the red cycle is one of the new cycles created by thickening (Fact~\ref{fact:thickened_cycles}). 
    \textbf{(d)} Cellulating cycles into triangles. We did not cellulate the red cycle(s) as they always have weight 4.}
    \label{fig:thicken_cellulate}
\end{figure}

To prove that the graph desiderata are satisfied, we need to analyze how relative expansion changes under the operations of thickening and cellulation.
Since cellulation only adds edges without adding vertices, it could only improve relative expansion. 
For thickening, we use the following lemma.
\begin{restatable}[Thickening Lemma]{lemma}{thickenlemma}\label{lem:thicken_expansion}
    Suppose $G = (V, E)$ has relative Cheeger constant $\beta = \beta_t(G, P)$ for port $P\subseteq V$ and integer $t$. 
    Fix $\ell\ge 1$. For all $r$ where $1\le r\le \ell$, we have
    \begin{equation}
        \beta_t(G\square J_\ell, P\times \{r\}) \ge \min(1, \ell\beta).
    \end{equation}
\end{restatable}
Similar versions of this lemma have been proved in Refs.~\cite{cross2024improved, williamson2024low,swaroop2024universal}. 
We include a proof in Appendix~\ref{apdx:proofs}.

\noindent \textbf{Proof of Lemma~\ref{lem:measurement_graph_construction}.}
We count the degree of $G$, and thereby bound the number of edges. 
The graph is initially empty. 
After Step~\ref{construction:matching}, every $v\in V_1$ has degree at most $\Delta$, because every stabilizer $S$ acting on $f^{-1}(v)$ adds at most one edge to $v$. 
After Step~\ref{construction:expand}, the max degree in $G_1$ is at most $\Delta + \delta$, and the number of edges is at most $|E_1|\le (\Delta+\delta)|L|/2$.
In Step~\ref{construction:decongestion}, we compute a cycle basis $\MR$ of congestion $\rho$ of $G_1$ and partition it into $t$ non-overlapping sets, $\MR = \bigcup_{i=1}^t \MR_i$.
By thickening in Step~\ref{construction:thicken}, the total number of edges in $G$ is at most 
\begin{align}
    |E| 
    &\le \ell\cdot |E_1| + (\ell-1)\cdot |V_1|.
\end{align}
In Step~\ref{construction:cellulation}, observe that when we cellulate a cycle $C$ of length $w$, we add $w-3$ edges and add degree at most 2 to any vertex in the cycle.
The cycles we cellulate on each level is non-overlapping, which means they have at most $|E_1|$ edges in total. 
Therefore, cellulation adds at most $|E_1|\cdot t\le |E_1|\cdot \ell$ edges to $G$, and the final number of edges is at most
\begin{align}\label{eqn:space_overhead}
    |E| 
    &\le 2\ell\cdot |E_1| + (\ell-1)\cdot |V_1|
    \le \left[\ell(\Delta+\delta+1)\right] |L|.
\end{align}
For a vertex $v\times \{r\}$ in any level $r$ of $G$, its degree before cellulation is at most $\Delta+\delta+2$, where two of its edges are connecting to its copies in the previous and next level.
Since the cycles in $\MR_r$ we cellulate are non-overlapping, $v\times \{r\}$ can be in at most $\frac{\Delta+\delta}{2}$ cycles, which means cellulation adds degree at most $\Delta+\delta$ to any vertex. 
We see that the maximum degree in $G$ is at most $2(\Delta+\delta+1)$.

By construction, we see that $G$ satisfy Desideratum~\ref{des:connected} and~\ref{des:degree}.
Let $\MA_r$ denote the set of triangles on level $r$ we obtain from cellulating the cycles in $\MR_r\times \{r\}$.
The set of cycles of $G$ which we measure is
\begin{equation}
    \MT \cup \MA_1 \cup \cdots \cup \MA_\ell.
\end{equation}
This is a cycle basis of $G$ by Fact~\ref{fact:thickened_cycles} and the fact that cellulation  breaks one cycle into a sum of triangles over $\FF_2$.
This set of cycle has congestion $2$, and every cycle has length $3$ or $4$, which satisfy desideratum~\ref{des:cycle_basis}.
Moreover, desideratum~\ref{des:path_matching} is satisfied by the perfect matchings constructed in Step~\ref{construction:matching}.

For desideratum~\ref{des:relative_expansion}, at step~\ref{construction:expand} we have that the graph $G_1$ has Cheeger constant $\beta_{G_1}\ge 1/\beta$, which means $\beta_{|L|}(G_1, V_1)\ge 1/\beta$.
By the Thickening Lemma (Lemma~\ref{lem:thicken_expansion}), after step~\ref{construction:thicken} we have $\beta_{|L|}(G, V_1\times \{1\})\ge 1$.
Since cellulation can only increase expansion, and $|L|\ge d$, we see that $\beta_d(G, f(L))\ge 1$.
\qed

\begin{remark}[Practical Considerations]
    It is important to note that this construction suffices for a theoretical proof of fault-tolerance, but is excessive for practical purposes.
    As mentioned in Remark~\ref{rmk:practical_conditions}, in practice many aspects of the desiderata can be relaxed. 
    We discuss techniques and considerations for practical constructions in Section~\ref{sec:toolkit_practical}.
\end{remark}

As discussed in Remark~\ref{rmk:multiple_code_blocks}, Lemma~\ref{lem:measurement_graph_construction} enables us to construct measurement graphs for arbitrary logical operators $\ML$ supported on one or many code blocks. 
Nonetheless, it is natural and beneficial to consider a more modular approach. 
Suppose we have the measurement graphs of two logical operators $\ML_1, \ML_2$. 
Can we connect them, with low cost, into a single measurement graph for the logical operator $\ML_1\ML_2$?
This problem was considered in Ref.~\cite{cross2024improved} (without the auxiliary graph framework) and later in Ref.~\cite{swaroop2024universal}.
Collectively, the two works developed the following solution, which we refer to as the \textbf{bridge/adapter system}. 
The same system are given two names for different use cases, as we explain later in this section.

\begin{definition}[Bridge/Adapter]\label{def:bridge}
    Given two graphs $G_1 = (V_1, E_1)$ and $G_2 = (V_2, E_2)$, a \textbf{bridge/adapter} between $G_1$ and $G_2$ is a set of non-overlapping edges $B$ between vertices $V_1$ and $V_2$.
\end{definition}

\begin{lemma}[Bridged Expansion]\label{lem:bridged_expansion}
    Suppose $G_1 = (V_1, E_1)$ and $G_2 = (V_2, E_2)$ have relative expansion $\beta_{t_1}(G_1,P_1)\ge 1$ and $\beta_{t_2}(G_2,P_2)\ge 1$ with respect to ports $P_1, P_2$.
    Let $B$ be a bridge between $P_1,P_2$. 
    The bridged graph $G = (V_1\cup V_2, E_1\cup E_2\cup B)$ has relative expansion $\beta_t(G, P_1\cup P_2)\ge 1$ for $t = \min(t_1,t_2,|B|)$.
\end{lemma}
We refer readers to Lemma~9 of Ref.~\cite{swaroop2024universal} for a proof.
Evidently, adding a bridge $B$ of edges between two disconnected graphs also creates $|B|-1$ new basis cycles in the bridged graph, which we need to measure as cycle checks.
If these cycles are high weight or have high congestion, the bridged system may not be LDPC.
The following lemma shows that given two measurement graphs, we can always find a bridge system that induces a desirable basis of cycles. 

\begin{lemma}[Sparsity of Bridge~\cite{swaroop2024universal}]
    \label{lem:sparse_bridge}
    Consider two graph $G_1 = (V_1, E_1)$ and $G_2 = (V_2, E_2)$ with vertex subsets $S_1, S_2$, such that $S_1, S_2$ induce connected subgraphs in $G_1,G_2$ respectively.
    Suppose $G_1,G_2$ have cycle bases with congestion $\rho$ and maximum length $\gamma$. 
    For any integer $b\le \min(|S_1, S_2|)$, we can efficiently find a bridge $B$ of size $b$ between $S_1, S_2$ such that the joined graph $G = (V_1\cup V_2, E_1\cup E_2\cup B)$ has a cycle basis with congestion at most $\rho+2$ and maximum length $\max(\gamma, 8)$.
\end{lemma}
This lemma is proved as Lemma~10 in Ref.~\cite{swaroop2024universal}, using a novel SkipTree algorithm.\footnote{It is not known whether this algorithm, and therefore Lemma~\ref{lem:sparse_bridge}, can be extended to arbitrary hypergraphs.}
Combining the two lemmas above, we have the desired primitive which connects two (and more) measurement graphs into bigger measurement graphs.
\begin{restatable}[Bridging Lemma~\cite{swaroop2024universal}]{lemma}{BridgingLemma}\label{lem:bridge}
    Suppose $G_1 = (V_1, E_1), G_2 = (V_2, E_2)$ and $f_1: L_1\rightarrow P_1, f_2: L_2\rightarrow P_2$ satisfy the graph desiderata for operators $\ML_1, \ML_2$ with non-overlapping support ($L_1\cap L_2 = \varnothing$). 
    We can efficiently compute a bridge $B$ of $d$ edges between $P_1$, $P_2$, which connects $G_1,G_2$ into the graph $G = (V_1\cup V_2, E_1\cup E_2\cup B)$.
    Let $f: L_1\cup L_2\rightarrow P_1\cup P_2$ be the port function where $f(q) = f_i(q)$ for $q\in L_i$.
    Then $G$ and $f$ satisfy the graph desiderata for the product operator $\ML_1\ML_2$.
\end{restatable}
We include a proof of this lemma, which is essentially the proof of Theorem~11 in Ref.~\cite{swaroop2024universal}, in Appendix~\ref{apdx:proofs} for completeness.
\begin{remark}\label{rmk:bridge_nuances}
    While the above formulation of Lemma~\ref{lem:bridge} suffices for the purpose of this paper, we note that there are many more nuanced ways to use bridges.
    For instance, in the lemma we assumed that the operators $\ML_1,\ML_2$ have non-overlapping support, which is not strictly necessary.
    In Section~3.6 and~3.7 of Ref.~\cite{cross2024improved},
    a bridge was added between the measurement hypergraphs of an $X$ operator and an anti-commuting $Z$ operator to perform a $Y$ measurement.
    In Ref.~\cite{swaroop2024universal}, Lemma~\ref{lem:bridge} is further stated in terms of \textit{sparsely overlapping} operators. 
    We suggest readers simply treat the bridge/adapter system as a modular approach to building measurement graphs, where the precise analysis can be done on a case-by-case basis.
\end{remark}

\begin{remark}\label{rmk:repeated_bridging}
    Evidently, Lemma~\ref{lem:bridge} can be applied iteratively to connect many measurement graphs into one global measurement graph. 
    This iterative application plays a crucial part in our later construction of the extractor architecture. 
\end{remark}

To distinguish between the names ``bridge'' and ``adapter'', in this work we refer to systems which connect measurement graphs for logical operators on the same code block or between blocks of the same code as \textbf{bridges}, and systems which connect measurement graphs for logical operators on different code families \textbf{adapters}. 
As discussed earlier, the adapters enable us to construct a universal fault-tolerant architecture based on different QLDPC code families. 
In comparison, prior code-switching schemes~\cite{bombin2016dimensional,breuckmann2017hyperbolic,Vasmer2019three} are all built with structurally analogous codes.

We summarize our toolkit with the following theorem, which is a straightforward combination of Theorem~\ref{thm:graph_desiderata}, Definition~\ref{def:measurement_protocol}, Theorem~\ref{thm:fault_distance} and Lemma~\ref{lem:measurement_graph_construction}.
\begin{theorem}[Ancillary Graph Surgery]\label{thm:ancillary_graph_surgery}
    Let $\ML$ be a logical operator supported on a LDPC code $\qcode$ (which could be made of multiple codes, see Remark~\ref{rmk:multiple_code_blocks}) with distance $d$. 
    Using $O(|L|(\log |L|)^3)$ ancilla qubits, we can construct another LDPC code $\bar{\qcode}$ such that code-switching between $\qcode$ and $\bar{\qcode}$ using the protocol of Theorem~\ref{thm:fault_distance} performs a logical measurement of $\ML$ on $\qcode$. 
    This protocol has fault distance $d$. 
\end{theorem}

\subsection{Surgery Toolkit: Practical Considerations}\label{sec:toolkit_practical}

The toolkit we have presented in this section is a theoretical blueprint, intended as a guide for surgery constructions in practice. 
However, it is important to note that almost every theoretical condition and proof we have written down is a loose upper bound.
In this section, we revisit the toolkit we have developed and discuss all our techniques from the practical perspective.

We start with Definition~\ref{def:graph_and_code}, which is our construction of measurement code $\qcode(\ML, G, f)$ from measurement graph $G$ and port function $f$. 
The purpose behind measuring a basis of cycle checks in $G$ is to ensure that $\bar{\qcode}$ does not have new, gauge logical qubits.\footnote{On a related note, this is why the scheme from Ref.~\cite{cross2024improved} was named the gauge-fixed surgery scheme -- all the gauge logical qubits were measured as stabilizers.} 
When these gauge qubits are not measured (or gauge-fixed), their operator can multiply with (or dress) the existing, unmeasured logical operators, lowering their weight and thereby the merged code distance. 
This is a problem observed in the CKBB scheme~\cite{cohen2022low} and one of the primary reasons behind their $O(d^2)$ space overhead of surgery.
Two notes are in order regarding measuring cycle checks in practice:
\begin{enumerate}[itemsep = 0pt]
    \item Some of these cycle checks may be redundant. More precisely, a cycle check $B_C$ may be a product of deformed stabilizers in the code $\qcode$ (Step~\ref{stab:modified_code_checks}). This is observed in constructions of ancilla systems for the $[[144,12,12]]$ bivariate bicycle code~\cite{cross2024improved,williamson2024low}, as well as the CKBB measurement hypergraphs on certain families of hypergraph product codes (Section~3.4 of~\cite{cross2024improved}).
    In these cases, the cycle check $B_C$ does not need to be explicitly measured. 

    \item Some gauge logical qubits do not hurt distance; in other words, the merged code distance may still be preserved with some cycle checks left out. This is observed in Ref.~\cite{cowtan2024ssip}, where various small-to-medium scale QLDPC codes are augmented by CKBB ancilla systems with number of levels less than $d$ (without measuring all cycle checks), yet their distances are still numerically estimated to be preserved.
\end{enumerate}
 
From Lemma~\ref{lem:measurement_graph_construction},
we see that the dominating factor in our space overhead comes from thickening which aims to decongest the cycles and potentially increase relative expansion. We again have several remarks in order.

\begin{enumerate}[itemsep = 0pt]
    \item Given the cycle checks we need to measure to preserve distance, there are many techniques one could apply to lower the overall congestion/overlap of these cycles. For instance, if too many cycles traverse an edge $e$, we can create a copy $e'$ and re-route half the cycles through $e'$, while adding a new cycle $(e,e')$.
    We note that finding cycle basis of low congestion or overlap is an interesting combinatorial optimization problem that to the best of our knowledge, has not been broadly studied beyond the decongestion lemma of Ref.~\cite{freedman2021building}.

    \item The intuition behind decongestion by thickening is to create more space so that the overlapping cycles can be spread out. Evidently, thickening the entire graph is conceptually simple yet practically unthrifty. In practice there are other techniques one could consider, such as thickening a subgraph.

    \item Relative expansion, as in desideratum~\ref{des:relative_expansion} of Theorem~\ref{thm:graph_desiderata}, is not strictly required for the merged code to preserve distance. More precisely, having (relative) Cheeger constant at least 1 is convenient for theoretical proofs but excessive in practice. In the merged code $\bar{\qcode} = \qcode(\ML, G, f)$, suppose $\ML$ is a $Z$ operator and $\ML'$ is another $Z$ operator that has overlapping support with $\ML$, $L'\cap L\ne \varnothing$. For $\bar{\qcode}$ to have distance $d$, we only need the vertex checks in $f(L'\cap L)$ to have large boundaries, not to have $\beta_d(G,P)\ge 1$ as in our constructions.
    It was observed in the 103-qubit system of Ref.~\cite{cross2024improved} that the measurement graphs were not expanding. Similarly, Ref.~\cite{williamson2024low} constructed a 41-qubit system on the $[[144,12,12]]$ BB code\footnote{Note that the 41-qubit system is a single measurement graph on one logical operator, while the 103-qubit system consists of two measurement graphs on four logical operators.} which is also not strictly expanding.
    These examples were nevertheless all proven to be distance preserving via integer programming. 
\end{enumerate}

With these in mind, we note that Lemma~\ref{lem:measurement_graph_construction} constructed a measurement graph with $O((\log |L|)^3)$ levels, while the 103-qubit system in Ref.~\cite{cross2024improved} and the 41-qubit system in Ref.~\cite{williamson2024low} both used only one level.
We expect this stark contrast between theoretical upper bounds and practical optimizations to persist in other QLDPC codes. 
When working with a specific code, one should consider co-designing the measurement graph given the code structures, or simply adding random edges as suggested in Ref.~\cite{williamson2024low}, or using a greedy algorithm as in Ref.~\cite{ide2024fault}.
Consequently, the connectivity overhead tracked in Lemma~\ref{lem:measurement_graph_construction} is also an overestimate.
Refs.~\cite{cross2024improved,williamson2024low} both accounted for the overall degrees of the merged code; Ref.~\cite{ide2024fault} also provided many detailed examples and discussions.

For time overhead, while Theorem~\ref{thm:fault_distance} asked for $3d$ rounds of syndrome measurement per logical measurement, in practice this should be based on benchmarking. In the 103-qubit system for the $[[144,12,12]]$ BB code, it was observed that 7 rounds of syndrome measurement balances the memory error rate and measurement error rate (Figure~10b of Ref.~\cite{cross2024improved}) on the distance 12 merged code. 
Decoding these syndrome information is yet another interesting question. In Ref.~\cite{cross2024improved}, a modular decoder was developed which can handle the syndrome information from the ancilla system and code $\qcode$ separately, which significantly improved the decoding time on the 103-qubit system (Figure~11 of Ref.~\cite{cross2024improved}).
Extending this result to the auxiliary graph surgery setting would be productive.
We further discuss the possibility of single-shot QLDPC surgery at the end of Section~\ref{sec:compilation}.

Overall, we believe the theoretical analysis in the preceding sections does not capture the overhead one would obtain through practical optimizations.

\section{Extractor Systems}\label{sec:extractors}

Prior works on QLDPC surgery, including the toolkit we presented above, are primarily concerned with constructing an ancilla system to measure a single, or a specified set of, logical operator(s). 
This approach has the inherent property of being addressable: the enabled logical measurements target specific logical qubits.
Such a fine-grained level of logical control is both amenable to implementation of generic algorithms, and uncommon among schemes of logical operations on QLDPC codes (see also the discussion in Section~1 of Ref.~\cite{he2025quantum}).
Nonetheless, this inherent addressability also comes with a notable downside: to measure different logical operators, or simply different bases of the same logical operators, we need to use different ancilla systems. 
This poses a considerable challenge in applying QLDPC surgery in practice.
Naively, if we build many ancilla systems to support a large family of logical measurements, the space and connectivity overhead could quickly become impractical.
If we consider rearranging physical qubit connectivity every time we perform a logical measurement, such as is possible in principle on a system of neutral atoms or shuttleable ions, the cost of rearrangement will incur a heavy time overhead. 
The proposed ancilla system in Ref.~\cite{cross2024improved} on the $[[144,12,12]]$ bivariate bicycle codes, which totals to 103 physical qubits,\footnote{This number accounts for both data and check qubits in the ancilla system.} can realize the full Clifford group on 11 out of 12 logical qubits when aided by automorphism gates. 
While it has a reasonable space and connectivity overhead, and does not require rearrangments of qubit connectivity, the automorphism gates again incur a time overhead. 
These challenges and tradeoffs call for a more wholistic approach to the design of such surgery ancilla systems. 

In this paper, we propose a construction of a single ancilla system which, when attached to the base code $\qcode$, enables logical measurement of any Pauli operator $\ML$ supported on the $k$ logical qubits of $\qcode$.
The ancilla system preserves the LDPC property of the base code $\qcode$.
We call this system a \textbf{single-block extractor}, or \textbf{extractor} for short. 
As we show, by joining multiple extractors together using bridge systems (Definition~\ref{def:bridge}), we can build a many-block fault-tolerant QLDPC architecture which, when supplemented by any magic state factory, is capable of performing universal Pauli-based computation.

\subsection{Single-block Extractor: Auxiliary Graph}\label{sec:single_block_extractor}

Consider an arbitrary quantum LDPC code $\qcode$. 
To build an extractor system on $\qcode$, we define the following extractor desiderata, which are similar to the graph desiderata of Theorem~\ref{thm:graph_desiderata}, with an important difference: the graph properties do not be dependent on any specific logical operator $\ML$; instead, its properties depend on $\qcode$ itself. 
\begin{definition}[Extractor Desiderata]
    \label{def:extractor_desiderata}
    Let $\qcode$ be a $[[n, k, d]]$ quantum code with physical qubits $Q$.
    Enumerate the stabilizer checks of $\qcode$ as $\MS = \{S_1, \cdots, S_m\}$.
    Consider a graph $\EXTG = (V, E)$ and an injective port function $F: Q\rightarrow V$.
    We refer to the following conditions on $\EXTG$ and $F$ as \textbf{extractor desiderata}.
    \begin{enumerate}[itemsep = 0pt]
        \item $\EXTG$ is connected.
        \label{ext-des:connected}

        \item \label{ext-des:LDPC} 
        \begin{enumerate}[itemsep = 0pt]
            \item The maximum degree of $\EXTG$ is $O(1)$; \label{ext-des:degree}

            \item There is a cycle basis $\MR$ of $\EXTG$ such that $\MR$ has congestion $O(1)$, and every cycle in $\MR$ has length $O(1)$. \label{ext-des:cycle_basis}

            \item There exists a collection of edge sets $\ME = \{E_1, \cdots, E_m\}$, $E_i\subseteq E$, such that 
            \begin{enumerate}
                \item For any even subset of qubits $K_i\subseteq Q$ in the support of $S_i$, there exists a path matching $\mu_i\subseteq E_i$ of $F(K_i)$. \footnote{Recall the notion of path matching from Definition~\ref{def:path_matching}.}
                \item Every $E_i$ has $O(1)$ edges and every edge in $E$ is in $O(1)$ sets $E_i$. 
            \end{enumerate}

            \label{ext-des:path_matching}
        \end{enumerate}

        \item $\beta_d(\EXTG, F(Q))\ge 1$. \label{ext-des:relative_expansion}
    \end{enumerate}
\end{definition}
The purpose behind this definition is evident from the following lemma.

\begin{lemma}[Extractor Lemma]
    \label{lem:extractor}
    Suppose $\EXTG, F$ satisfy the extractor desiderata of Definition~\ref{def:extractor_desiderata} with respect to code $\qcode$ supported on qubits $Q$. 
    Let $\ML$ be a logical operator of $\qcode$ with support $L\subseteq Q$. 
    Let $f_L = F\vert_L$, namely, the function $F$ with domain restricted to $L$. 
    Then $\EXTG, f_L$ satisfy the graph desiderata of Theorem~\ref{thm:graph_desiderata}.
\end{lemma}
\begin{proof}
    The graph desiderata~\ref{des:connected},~\ref{des:degree},~\ref{des:cycle_basis} are the same as the respective extractor desiderata.
    Graph desideratum~\ref{des:path_matching} reduces to extractor desideratum~\ref{ext-des:path_matching} because for every check $S_i$, there is a path matching $\mu_i\subset E_i$ of $F(K(S, \ML))$ (recall that $K(S, \ML)$ is the set of qubits on which $S$ and $\ML$ anti-commutes).
    Graph desideratum~\ref{des:relative_expansion} reduces to extractor desideratum~\ref{ext-des:relative_expansion} by the Restriction Lemma~\ref{lem:restriction_relative_expansion}.
\end{proof}
The direct consequence of this lemma is that if we can construct a graph satisfying the extractor desiderata, then we can use it to perform fault-tolerant logical measurement of any logical operator $\ML$ of $\qcode$. 
We elaborate on this point in Section~\ref{sec:single_block_computation}.
Here, we show how to construct such a graph, following a procedure similar to Lemma~\ref{lem:measurement_graph_construction}.

\begin{lemma}[Extractor Construction]\label{lem:single_block_extractor}
    Let $Q$ denote the physical qubits of $\qcode$, enumerate $Q$ as $Q[1], \cdots, Q[n]$.
    \begin{enumerate}[itemsep = 0pt]
        \item Let $V_1 = \{v_1, \cdots, v_n\}$ be a set of vertices of size $n$.
        Define the port function $F(Q[i]) = v_i$.
        \item Construct a base graph $\EXTG_1 = (V_1, E_1)$ as follows: \label{extractor:base}
        \begin{enumerate}[itemsep = 0pt]
            \item For every stabilizer $S\in \MS$, suppose $S$ acts on qubits $Q[{s_1}], \cdots, Q[{s_\omega}]$. Add a cycle of $\omega$ edges with vertices $v_{s_1}, \cdots, v_{s_\omega}$ to $\EXTG_1$. 
            Denote this cycle $C(S)\subset E_1$.
            \label{extractor:stabilizers}
            \item Construct a constant degree graph $D$ on $n$ vertices with Cheeger constant $\beta_D\ge \beta$ for some constant $\beta$ of our choice. Add the edges of $D$ to $\EXTG_1$.
            \label{extractor:expand}
        \end{enumerate}
        \item \label{extractor:decongestion} Apply the Decongestion Lemma (Lemma~\ref{lem:decongestion}) and Corollary~\ref{cor:decongestion} to obtain a cycle basis of $\MR$ with congestion $\rho$ of $\EXTG_1$, 
        and a partition $\MR = \bigcup_{i=1}^t \MR_i$ such that each $\MR_i$ contains non-overlapping cycles and $t\le O((\log n)^3)$.
        \item \label{extractor:thicken} Thicken $\EXTG_1$ by $\ell = \max(t, 1/\beta)$ times to obtain $\EXTG = \EXTG_1\square J_\ell$, denote $\EXTG = (V, E)$.
        \item \label{extractor:cellulation} On every level $\EXTG\times \{r\}$, cellulate every cycle in $\MR_r$ to obtain a collection of triangles which generate the cycles $\MR_r\times \{r\}$.
    \end{enumerate}
    Since $\qcode$ is a LDPC code, we assume every stabilizer $S$ has weight at most $\omega$ and every qubit in $Q$ is checked by at most $\Delta$ stabilizers. 
    Suppose the expander graph we used in Step~\ref{extractor:expand} has maximum degree $\delta$.
    Then the constructed graph $\EXTG$ has maximum degree at most $4\Delta+2(\delta+1)$ and total edges at most $\ell(2\Delta+\delta+1)n \le O(n(\log n)^3)$. 
    The cycle basis we measure for $\EXTG$ has congestion $2$ and maximum length $4$.
    $\EXTG$ and $F$ satisfy the extractor desiderata of Definition~\ref{def:extractor_desiderata}.
\end{lemma}
\begin{proof}
    Note that the same proof as for Lemma~\ref{lem:measurement_graph_construction} applies, except for changes to Step~\ref{extractor:stabilizers} and therefore extractor desideratum~\ref{ext-des:path_matching}.
    The sets $E_i$ required are precisely the cycles $C(S)$ we added; they have size bounded by $\omega$ and do not overlap.
    For each $K_i$, since $|K_i|$ is even, we can find a path matching $\mu_i$ for $K_i$ inside the cycle $C(S_i)$. 
    These matchings have weight at most $\omega/2$.
    Since we add a cycle of edges for every stabilizer, and every qubit in $Q$ is checked by at most $\Delta$ stabilizers, the degree of vertices in $\EXTG_1$ after Step~\ref{extractor:base} is at most $2\Delta+\delta$.
    The rest of the arguments follow as in Lemma~\ref{lem:measurement_graph_construction}.
\end{proof}

As in the case of measurement graph desiderata and construction, our extractor desiderata and construction suffice for theoretical analysis, but are excessive for practical purposes.
We discuss techniques and considerations for practical use of extractors in Section~\ref{sec:extractor_practical}.

\subsection{Single-block Extractor: Computation System with Fixed Connectivity}\label{sec:single_block_computation}

We now discuss how to build a QLDPC computational block from $\qcode$ and extractor graph $\EXTG$, which encodes the $k$ qubits from $\qcode$ and support logical measurement of any logical operator $\ML$ supported on these logical qubits.
More precisely, we describe a system of data and check qubits with fixed connectivity and show that this system can implement all measurement codes $\qcode(\ML, \EXTG, F\vert_L)$ (recall Definition~\ref{def:graph_and_code}).

\begin{figure}
    \centering
    \includegraphics[width = \textwidth]{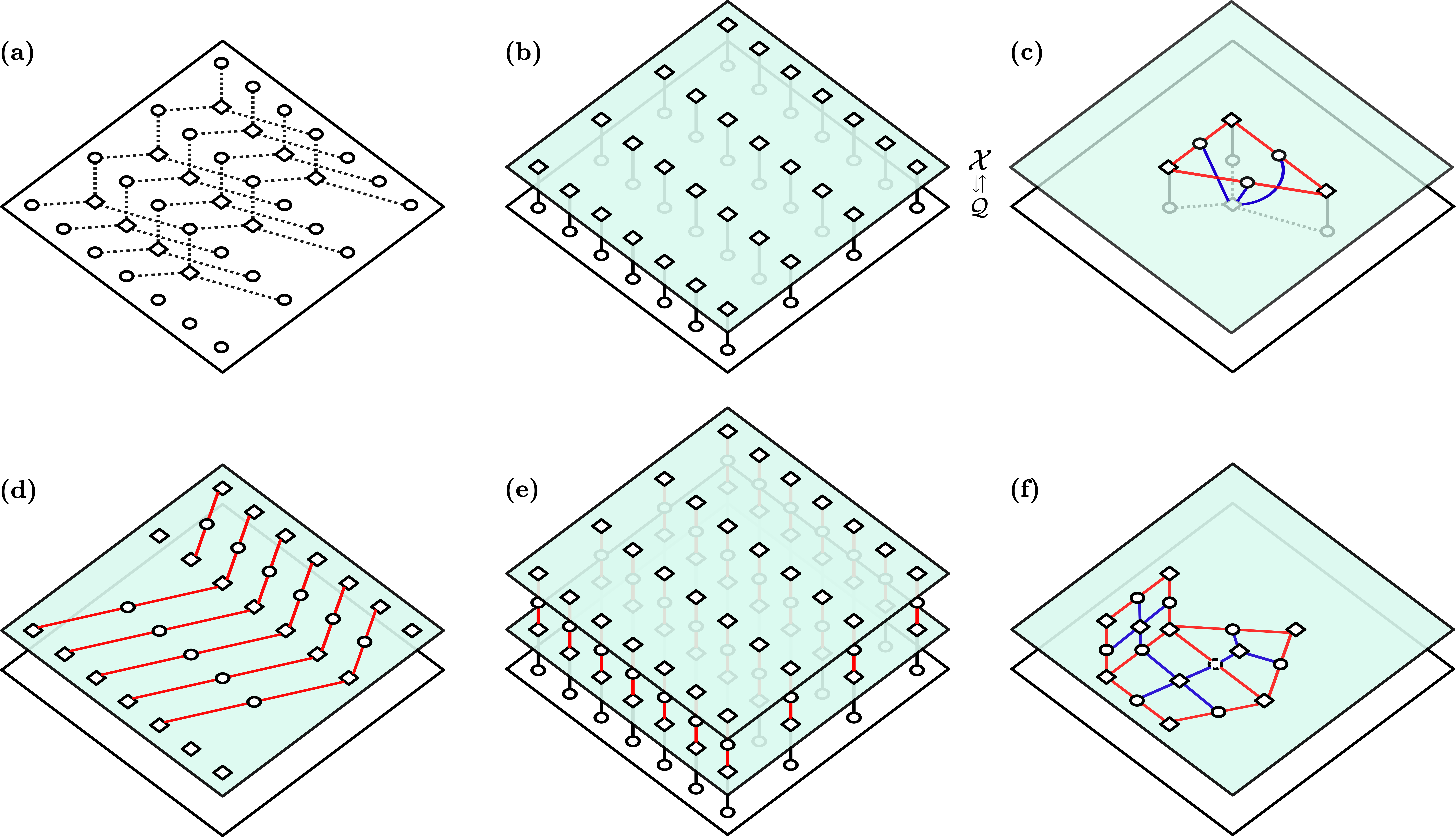}
    \caption{Depiction of different components of an Extractor-Augmented Computation (EAC) block. In all panels, circles denote data qubits, sqaures denote check qubits, and lines denote connections between check and data qubits. The color of a line denote the Pauli action of the check qubit on the data qubit: red for $Z$, blue for $X$. 
    \textbf{(a)} A generic quantum CSS code $\qcode$ with data and check qubits. We use dotted lines to indicate that these connections came with the code. The lines are uncolored as their Pauli actions are unspecified. 
    \textbf{(b)} The first level of an extractor $\EXT$ (light green) has one check qubit per data qubit in $\qcode$. They are connected 1-to-1. The lines are uncolored as their Pauli actions are unspecified.
    \textbf{(c)} For every stabilizer $S$ of $\qcode$, we add a cycle $C$ of edges among the vertex checks that are connected to the qubits in support of $S$. Each edge is a data qubit in $\EXT$. The vertex checks act on edge qubits by $Z$. The stabilizer $S$ are extended to act on the edge qubits by $X$. Together, panels b, c depict all the coupling edges $\leftrightarrows$ between $\qcode$ and $\EXT$.
    \textbf{(d)} The base graph (first level) of an extractor $\EXTG_1$ is a constant degree expander graph. In practice, due to the underlying code structure, one should consider co-designing the extractor graph with the code, see Section~\ref{sec:extractor_practical}.
    \textbf{(e)} An extractor may have multiple levels due to thickening (Definition~\ref{def:thicken}). 
    \textbf{(f)} Create cycle checks for a basis of cycles of $\EXTG$. We depict two types of cycles here: the vertical cycle between levels of $\EXTG$ comes from thickening (Fact~\ref{fact:thickened_cycles}), and the horizontal cycles on each level come from $\EXTG_1$. Cycle checks act on edge qubits by $X$. The circuit with dashed boundary denote a qubit that came from cellulation (Definition~\ref{def:cellulation}).
    }
    \label{fig:extractor}
\end{figure}

\begin{definition}[Extractor System]\label{def:extractor_system}
    Let $\EXTG = (V, E), F$ be an extractor graph and port function satisfying the extractor desiderata (Definition~\ref{def:extractor_desiderata}) for code $\qcode$. 
    Here, we construct a fixed system of data and check qubits based on $\qcode, \EXTG$ and $F$.
    For clarity of notation, we denote data qubits by $Q$ and $Q_\EXTG$, 
    and check qubits by $H$ and $H_\EXTG$.
    \begin{enumerate}[itemsep = 0pt]
        \item Enumerate the data qubits of $\qcode$ as $Q[1], \cdots, Q[n]$, and the check qubits as $H[S]$ for stabilizers $S\in \MS$. Connect every check qubit $H[S]$ to the data qubits in the support of $S$.
        See Figure~\ref{fig:extractor}a.
        \item For every vertex $v\in V$, create a check qubit $H_\EXTG[v]$. Denote $v_i := F(Q[i])$, connect $H_\EXTG[v_i]$ to $Q[i]$. See Figure~\ref{fig:extractor}b.
        \item For every edge $e\in E$, create a data qubit $Q_\EXTG[e]$. 
        Suppose $e = (u, v)$, connect $Q_\EXTG[e]$ to $H_\EXTG[u], H_\EXTG[v]$. 
        See Figure~\ref{fig:extractor}c,d,e,f.
        \item For every cycle $C$ in the cycle basis $\MR$, create a check qubit $H_\EXTG[C]$, connect it to all the edge qubits $Q_\EXTG[e]$ for $e\in C$.
        See Figure~\ref{fig:extractor}f.
        \item For every stabilizer $S_i\in \MS$, connect $H[S]$ to $Q_\EXTG[e]$ for all $e\in E_i$ (recall extractor desideratum~\ref{ext-des:path_matching}), see Figure~\ref{fig:extractor}c.
    \end{enumerate} 
    We denote the anxillary system made of $Q_\EXTG$ and $H_\EXTG$ as $\EXT$, and the full system as $\fullsystem$.
\end{definition}

\begin{theorem}[Extractor-Augmented Computation Block]\label{thm:EAC_block}
    For any QLDPC code $\qcode$ with parameters $[[n,k,d]]$, let us construct 
    an extractor graph $X$ and port function $F$ and consider the joined system $\fullsystem$ from Definition~\ref{def:extractor_system}.
    \begin{enumerate}[itemsep = 0pt]
        \item $\fullsystem$ is LDPC, and the total number of qubits (data and check) is bounded by $O(n(\log n)^3)$. 
        \label{EAC:space_bound}
        \item $\fullsystem$ can implement all measurement codes $\qcode(\ML, X, F\vert_L)$, for any logical Pauli operator $\ML$ of $\qcode$ with support $L\subseteq Q$.
        \label{EAC:measurement}
    \end{enumerate}
    We refer to the full system $\fullsystem$ as an \textbf{Extractor-Augmented Computation block}, or an \textbf{EAC block} for short. 
    As a direct consequence of the auxiliary graph surgery scheme (Definition~\ref{def:measurement_protocol} and Theorem~\ref{thm:fault_distance}), we can implement the code $\qcode$ and perform logical measurement of any $\ML$ in the EAC block, without the need of SWAP gates or qubit connectivity rearrangements.
    One logical measurement step uses $O(d)$ syndrome measurement cycles and has fault distance $d$.
\end{theorem}
\begin{proof}
    Property~\ref{EAC:space_bound} is a direct consequence of the construction of extractor graph (Lemma~\ref{lem:single_block_extractor}) and extractor system (Definition~\ref{def:extractor_system}). 
    Property~\ref{EAC:measurement} follows from Lemma~\ref{lem:extractor}.
    The remaining claim on correctness and fault tolerance follows from the auxiliary graph surgery toolkit: Theorem~\ref{thm:graph_desiderata}, Definition~\ref{def:measurement_protocol} and Theorem~\ref{thm:fault_distance}.
\end{proof}

\begin{figure}
    \centering
    \includegraphics[width = 0.4\textwidth]{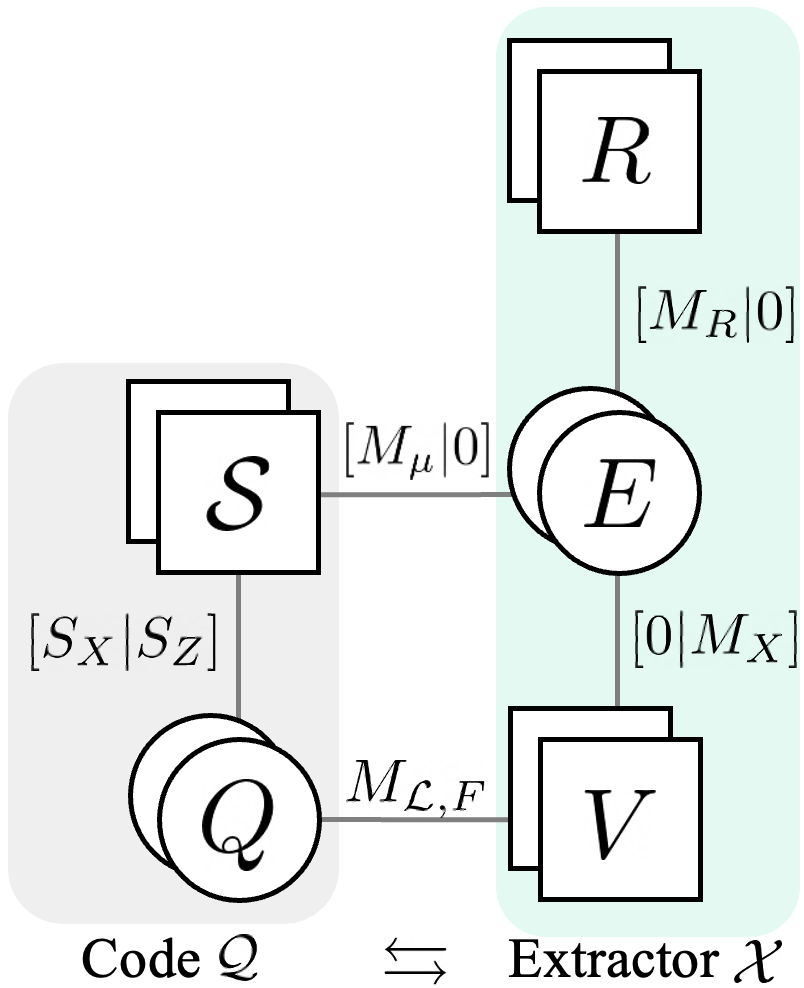}
    \caption{Logical measurement in an EAC block, depicted with scalable Tanner graphs. Similar to Figure~\ref{fig:gauging}, we have a code system $\qcode$ and an ancilla system based on the extractor graph $\EXTG$. 
    When measuring a logical operator $\ML$, the coupling connections $\leftrightarrows$ activated depend on $\ML$ and $F$. Specifically, for qubits $q\in L$, the vertex check $F(q)$ acts on $q$ by $\ML(q)$. This is captured as a sympletic matrix $M_{\ML, F}$. For every check $S$ with $K(S,\ML)\ne \varnothing$, $S$ acts by $X$ on edge qubits, subset of the cycle $C(S)$, which form a path matching of $K(S,\ML)$. This is captured as a matrix $M_{\mu}$.
    }
    \label{fig:extractor-measure}
\end{figure}

\begin{remark}[Uniformity in Logical Measurements]\label{rmk:uniformity}
    We elaborate on how an EAC block performs logical measurements to illustrate our fixed-connectivity approach.
    To measure a logical operator $\ML$, we use a code-switching protocol (Definition~\ref{def:measurement_protocol}, Theorem~\ref{thm:fault_distance}) between $\qcode$ and $\qcode(\ML, \EXTG, F\vert_L)$. 
    Since we use the same measurement graph $\EXTG$ for any $\ML$, the stabilizers $\MS(\ML, \EXTG, F\vert_L)$ for different $\ML$ are mostly identical. 
    More precisely, in Definition~\ref{def:graph_and_code}, the stabilizers that change between measurements of different logical operators are exactly the checks that connect the code system $\qcode$ to the extractor system $\EXT$: stabilizers~\ref{stabilizers:port-vtxs} and~\ref{stab:modified_code_checks}.
    In other words, the code $\qcode$ and extractor $\EXT$ systems remain uniform across measurements, while different sets of coupling edges $\leftrightarrows$ are activated for each $\ML$. 

    From the hardware perspective, such uniformity in our computational procedures substantially simplifies control sequences and reduces calibration overhead. 
    On the circuit level, activating different coupling edges in $\leftrightarrows$ corresponds to using different CNOT gates between $\qcode$ and $\EXT$, without changing the (partial stabilizer measurement) circuits on $\qcode$ or $\EXT$.
    For systems with fixed connectivity, such as superconducting devices, these changes can be easily implemented assuming the connectivity detailed in Definition~\ref{def:extractor_system} is built. 
    For systems which support qubit rearrangments, such as neutral atom devices, qubit movements for the partial circuits on $\qcode$ and $\EXT$ would remain unchanged while movements to couple $\qcode$ and $\EXT$ would be largely similar.

    This uniformity is further extended into our architectural proposal. In Section~\ref{sec:architectures}, we primarily consider what we call \textbf{uniform architectures}, where we connect many instances of the same EAC blocks with bridges for computation. This is in contrast to \textbf{hybrid architectures}, where we connect EAC blocks based on different QLDPC code families using adapters. 
    In a uniform architecture, any local optimizations to the system or the measurement scheme can be easily propagated into a global optimization.
\end{remark}

\subsection{Multi-block Extractors}~\label{sec:multi-block}

A single-block extractor $\EXT$ can augment a QLDPC memory block into a computation block supporting logical measurements. 
Using a bridge system (Definition~\ref{def:bridge}), two or more extractors can be connected together to enable Pauli-based computation on multiple QLDPC codeblocks. 
As in Remark~\ref{rmk:multiple_code_blocks}, while we can simply apply Lemma~\ref{lem:single_block_extractor} to construct an extractor on multiple codeblocks, using a bridge is more modular and requires less connectivity between (augmented) blocks.

\begin{restatable}[Bridging Extractors]{lemma}{BridgingExtractors}
    \label{lem:bridging_extractors}
    Suppose $\EXTG_1 = (V_1, E_1), \EXTG_2 = (V_2, E_2)$ and $F_1: Q_1\rightarrow P_1, F_2: Q_2\rightarrow P_2$ satisfy the extractor desiderata (Definition~\ref{def:extractor_desiderata}) for codes $\qcode_1, \qcode_2$, supported on disjoint qubits $Q_1, Q_2$. 
    We can efficiently compute a bridge $B$ of $d$ edges between $P_1$, $P_2$, which connects $\EXTG_1,\EXTG_2$ into the graph $\EXTG = (V_1\cup V_2, E_1\cup E_2\cup B)$.
    Let $F: Q_1\cup Q_2\rightarrow P_1\cup P_2$ be the port function where $F(q) = F_i(q)$ for $q\in Q_i$.
    Then $\EXTG$ and $F$ satisfy the extractor desiderata for the code $\qcode_1\cup \qcode_2$.
\end{restatable}
We include a proof of this lemma, which follows the proof of Lemma~\ref{lem:bridge}, in Appendix~\ref{apdx:proofs}.
As in Remark~\ref{rmk:bridge_nuances} and~\ref{rmk:repeated_bridging}, the bridge system can be used repeatedly and in many more nuanced ways, which should be explored when building extractors on specific codes.

\begin{figure}
    \centering
    \includegraphics[width = 0.6\textwidth]{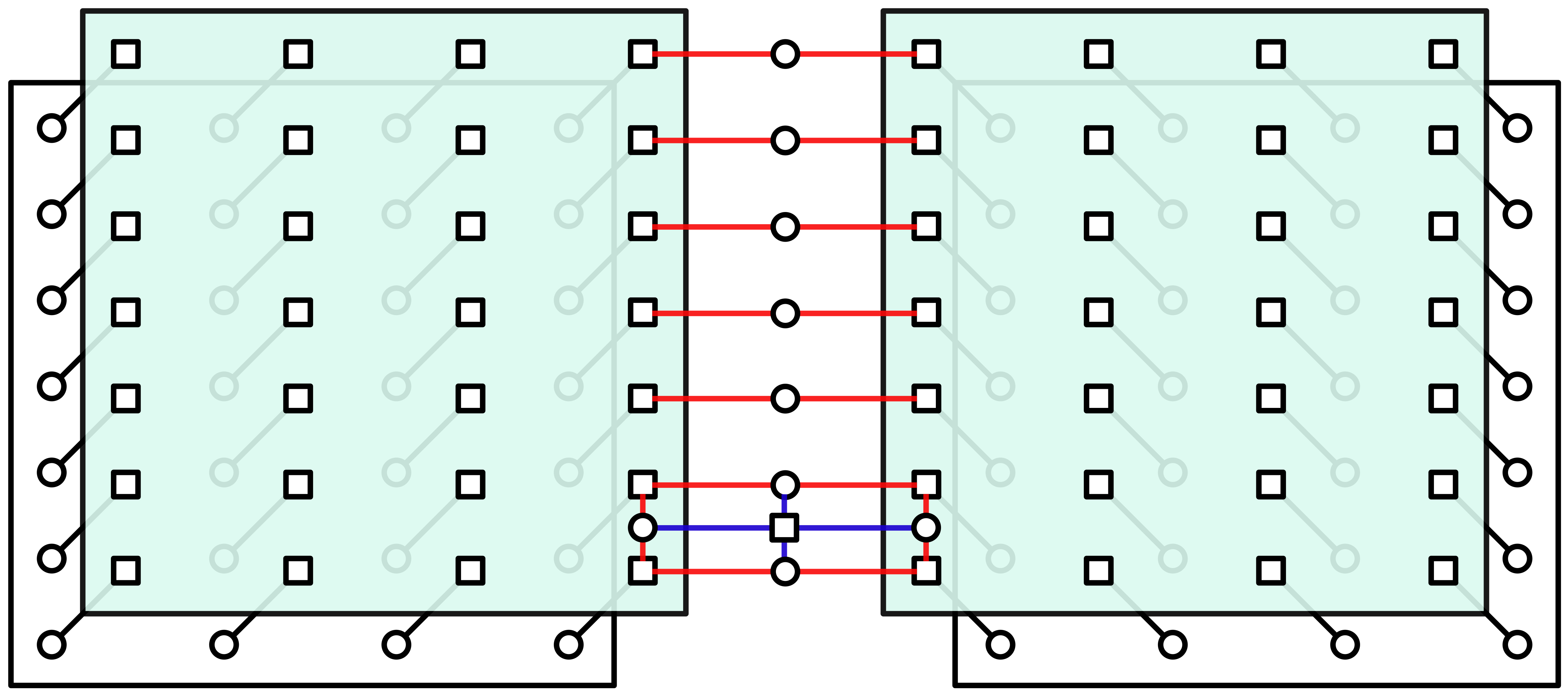}
    \caption{Two EAC blocks connected by a bridge system. A bridge system consists of a set of $d$ edges between the extractor graphs, and a collection of $d-1$ new cycle checks. Lemma~\ref{lem:sparse_bridge} guarantees that we can find a set of edges such that the new cycle checks are sparse. In this figure we depict one simple cycle check created by the bridge edges.
    }
    \label{fig:bridged_extractors}
\end{figure}

This simple construction produces a powerful primitive. 
Note that Lemma~\ref{lem:bridging_extractors}, much like the earlier versions of bridging lemmas, makes no assumption on the codes $\qcode_1,\qcode_2$. This means that by using an adapter, we can connect EAC blocks based on different QLDPC code families together into, pedantically, a larger EAC block.
We include a depiction in Figure~\ref{fig:bridged_extractors}.

\begin{corollary}[Multi-Block Measurements]\label{cor:multi-block-meas}
    Let $\qcode = \qcode_1\cup \qcode_2$. The system $\fullsystem$ can
    perform measurements of all operators of the form $\ML = \ML_1\ML_2$, where $\ML_1, \ML_2$ are any two operators in $\qcode_1,\qcode_2$, with $O(d)$ syndrome measurement cycles and fault distance $d$.
\end{corollary}

This primitive is one of the key foundations of our architectural proposal in Section~\ref{sec:architectures}.

\subsection{Partial Extractors}\label{sec:partial_extractors}

The above construction of single and multi-block extractors is both explicit and general, as Lemma~\ref{lem:single_block_extractor} and Definition~\ref{def:extractor_system} can augment any quantum LDPC code $\qcode$ into an EAC block $\fullsystem$ supporting logical measurements of all Pauli operators.
This generality enables one to construct and optimize extractor systems on promising code families, such as hypergraph product codes~\cite{tillich2013quantum,kovalev2012improved,leverrier2015quantum}, balanced product and lifted product codes~\cite{breuckmann2021balanced,panteleev2021degenerate,panteleev2021quantum}, bivariate bicycle codes~\cite{bravyi2024highthreshold}, and several more~\cite{wang2023abelian,lin2024quantum,malcolm2025computing}. 
Some of these codes admit various symmetry properties which enable low-depth implementation of subsets of logical Clifford gates~\cite{breuckmann2024foldtransversal,quintavalle2023partitioning,eberhardt2024operators,malcolm2025computing}. 
For these codes, building a full single-block extractor may be unnecessary, as a smaller (and therefore more limited) extractor could still be computationally versatile when combined with the native low-depth gates.
With these considerations in mind, in this section we define the notion of a \textbf{partial extractor} and discuss its usage.

Intuitively, a partial extractor is defined on a subset of qubits, which fully contains one or more logical operators. It is more than a measurement graph for single operators and less than a single-block extractor. 
\begin{definition}[Partial Extractor Desiderata]\label{def:partial_extractors_desiderata}
    Let $\qcode$ be a $[[n, k, d]]$ quantum code with physical qubits $Q$.
    Enumerate the stabilizer checks of $\qcode$ as $\MS = \{S_1, \cdots, S_m\}$.
    For a subset of qubits $T\subseteq Q$, consider a graph $\EXTG_T = (V, E)$ and an injective function $F_T: T\rightarrow V$.
    We say that $\EXTG_T$ and $F_T$ satisfy the \textbf{extractor desiderata with respect to $T$} if the following conditions hold
    \begin{enumerate}[itemsep = 0pt]
        \item $\EXTG_T$ is connected.
        \label{partial-des:connected}

        \item 
        \label{partial-des:LDPC} 
        \begin{enumerate}[itemsep = 0pt]
            \item The maximum degree of $\EXTG_T$ is $O(1)$;
             \label{partial-des:degree}

            \item There is a cycle basis $\MR$ of $\EXTG_T$ such that $\MR$ has congestion $O(1)$, and every cycle in $\MR$ has length $O(1)$. 
            \label{partial-des:cycle_basis}

            \item There exists a collection of edge sets $\ME = \{E_1, \cdots, E_m\}$, $E_i\subseteq E$, such that 
            \begin{enumerate}
                \item For any even subset of qubits $K_i\subseteq T$ in the support of $S_i$, there exists a path matching $\mu_i\subseteq E_i$ of $F_T(K_i)$. 
                \item Every $E_i$ has $O(1)$ edges and every edge in $E$ is in $O(1)$ sets $E_i$. 
            \end{enumerate}

            \label{partial-des:path_matching}
        \end{enumerate}

        \item $\beta_d(\EXTG_T, F_T(T))\ge 1$. 
        \label{partial-des:relative_expansion}
    \end{enumerate}
\end{definition}
Note that the full extractor desiderata (Definition~\ref{def:extractor_desiderata}) is precisely the above definition with $T$ set to $Q$.
Therefore, we can straightforwardly adapt the augmented system construction (Definition~\ref{def:extractor_system}) to obtain a \textit{$T$-augmented system} $\partialsystem{T}$, and adapt Theorem~\ref{thm:EAC_block} to the following form.

\begin{theorem}
    For any QLDPC code $\qcode$ with parameters $[[n,k,d]]$, for any subset of qubits $T\subseteq Q$, we can construct $\EXTG_T, F_T$ and a partially augmented system $\partialsystem{T}$ such that
    \begin{enumerate}[itemsep = 0pt]
        \item $\partialsystem{T}$ is LDPC, and the total number of qubits (data and check) is bounded by $O(|T|(\log |T|)^3)$. 

        \item $\partialsystem{T}$ can implement all measurement codes $\qcode(\ML, X_T, F_T\vert_L)$, for any logical Pauli operator $\ML$ of $\qcode$ with support $L\subseteq T$.

    \end{enumerate}
    One logical measurement step uses $O(d)$ syndrome measurement cycles and has fault distance $d$.
\end{theorem}
However, we warn that unlike full extractors, bridging two partial extractors may not give a larger partial extractor satisfying the above desiderata!
More precisely, given partial extractors and port functions $\EXTG_T, \EXTG_R, F_T, F_R$, joining $\EXTG_T, \EXTG_R$ at their bases $P_T = F_T(T), P_R = F_R(R)$ with a bridge gives a new graph $\EXTG_{T\cup R}$, which is guaranteed to satisfy all desiderata of Definition~\ref{def:partial_extractors_desiderata} except, unfortunately,~\ref{partial-des:path_matching}.
This is because for a stabilizer check $S$ with support in both $T$ and $R$, we may take qubits $t\in T\cap \supp(S)$, $r\in R\cap \supp(S)$, and it is not guaranteed that there is a short path between $t, r$ in $\EXTG_{T\cup R}$. 
Consequently, from a worst case point of view $\EXTG_{T\cup R}$ enables measurements of a subset of logical operators supported on $T\cup R$.

\begin{lemma}\label{lem:joined_partial_extractors}
    Let $\qcode$ be a $[[n,k,d]]$ quantum code with qubits $Q$, and let $\EXTG_T, \EXTG_R, F_T, F_R$ be partial extractor graphs and port functions satisfying the extractor desiderata with respect to $T, R\subseteq Q$.
    Suppose $T\cap R = \varnothing$, connect $\EXTG_T, \EXTG_R$ into $\EXTG_{T\cup R}$ using a bridge given by Lemma~\ref{lem:bridge}, and let $F_{T\cup R} = F_T\cup F_R$. 
    For any logical operator $\ML = \ML_1\ML_2$, wth $\ML_1$ supported on $L_1\subseteq T$ and $\ML_2$ supported on $L_2\subseteq R$, the graph $\EXTG_{T\cup R}$ and port function $F_{T\cup R}\vert_{L_1\cup L_2}$ satisfy the measurement graph desiderata (Theorem~\ref{thm:graph_desiderata}) for $\ML$. Note that one of $\ML_1, \ML_2$ could be identity.
\end{lemma}
\begin{corollary}\label{cor:bridged_partial_extractors}
    The partially augmented system $\partialsystem{T\cup R}$ constructed by Definition~\ref{def:extractor_system} using $\EXTG_{T\cup R}$ and $F_{T\cup R} = F_T\cup F_R$ can perform measurements of all operators of the form $\ML = \ML_1\ML_2$, where $\ML_1$ is supported on $L_1\subseteq T$ and $\ML_2$ is supported on $L_2\subseteq R$, with $O(d)$ syndrome measurement cycles and fault distance $d$.
\end{corollary}

The proofs are simple adaptations of proofs for Lemma~\ref{lem:bridging_extractors} and Theorem~\ref{thm:EAC_block}.
For operators $\ML$ that cannot be decomposed into a disjoint product, one could still try to measure them with $\EXTG_{T\cup R}$, but the theoretical guarantee on the LDPC property no longer holds.
To join two partial extractors into a larger partial extractor, which can measure such operators $\ML$ without breaking the LDPC property, one should add more edges to the bridged graph $\EXTG_{T\cup R}$ to satisfy desideratum~\ref{partial-des:path_matching}.

In retrospect, the 103-qubit ancilla system on the $[[144,12,12]]$ bivariate bicycle code constructed in Ref.~\cite{cross2024improved} can be seen as a partial extractor.
It is made of two auxiliary graphs, $G_X = (V_X, E_X)$ and $G_Z = (V_Z, E_Z)$, connected through a bridge. 
Vertices in $V_X, V_Z$ are each connected to two qubits in the gross code (instead of one, as standard in our constructions), which means we actually have two port functions for each graph. 
By activating coupling edges $\leftrightarrows$ corresponding to different port functions, this partial extractor system enables logical measurements of eight different operators (see also Table~1 of Ref.~\cite{cross2024improved}).
While the auxiliary graphs are not strictly speaking expanding (their Cheeger constants are less than 1), the merged code distances were verified to be 12 via integer programming. 
This design was improved in the bicycle architecture of Ref.~\cite{yoder2016universal}, where the ancilla systems are called logical processing units. 

These examples should be seen as a proof-of-concept: on a code with inherent constant-depth logical Clifford gates,
we can build a (small) partial extractor to significantly boost the computationally capability of our code, even though the partial extractor system may have limited measurement capability by itself.

Moreover, partial extractors can be especially useful for codes where logical operators can be found with disjoint supports, such as the bivariate bicycle codes~\cite{bravyi2024highthreshold} and hypergraph product codes~\cite{tillich2013quantum}, because we can build partial extractors on such disjoint supports and connect them with bridges.
As discussed above, it is important to make sure desideratum~\ref{ext-des:path_matching} is satisfied.

\subsection{Extractors: Practical Considerations}\label{sec:extractor_practical}

As in the case of auxiliary graph surgery, the extractor desiderata (Definition~\ref{def:extractor_desiderata} and~\ref{def:partial_extractors_desiderata}) and construction (Lemma~\ref{lem:extractor}) we discussed in this section are theoretical upper bounds that have immense room for practical optimizations. We refer readers to Section~\ref{sec:toolkit_practical} as the same ideas apply.
In this section, we focus more on optimizing the measurement and computational capacity of (partial) extractors.

We first consider full, single-block extractors.     
A natural question to ask is, does every operator require the use of the full extractor to be measured? In other words, for a logical operator $\ML$ of small physical support, can we use a subgraph $\EXTG'$ of $\EXTG$ to perform the measurement?
If $\EXTG'$ satisfy the graph desiderata for $\ML$, then evidently the measurement can be done fault-tolerantly.
However, that is not guaranteed. In practice one should check whether $\EXTG'$ is distance preserving, and potentially add more edges/qubits if not.
We note that using a subgraph of $\EXTG$ instead of the full extractor could potentially lead to a lower logical error rate, and require fewer syndrome measurement cycles per logical cycle.

A related idea is, can an extractor perform parallel measurements, similar to the parallel surgery schemes in Refs.~\cite{zhang2024time,cowtan2025parallel}?
Consider two commuting operators $\ML_1,\ML_2$. If $\ML_1$ and $\ML_2$ have overlapping support, then the current design of extractors cannot measure them simultaneously.
In the case where they have no overlap, let $P_1 = F(L_1), P_2 = F(L_2)$. 
We can find an edge cut in $\EXTG$ which separates the vertex sets $P_1, P_2$, and deactivate those edges (which corresponds to data qubits) in $\EXTG$. 
This effectively cuts an extractor $\EXTG$ into two subgraphs $\EXTG_1,\EXTG_2$, and the same discussions in the previous paragraph apply.

Next we discuss the construction and usage of partial extractors.
As discussed at the start of Section~\ref{sec:partial_extractors}, for QLDPC codes with desirable structure or low-depth Clifford gates, partial extractor(s) may be combined with such gates to reach the measurement capacity of a full extractor. 
This is precisely the case in the 103-qubit system for the gross code: while the partial extractor is only supported on 4 logical operators (2 logical operators and their ZX dual, see Section~9.1 of Ref.~\cite{bravyi2024highthreshold}), we can conjugate the set of measurable operators with the automorphism gates on the code, so that a larger set of logical operators supported on the 12 logical qubits can be measured, sufficient to perform 11-qubit Clifford computation.
Evidently, this approach incurs a compilation time overhead, which one should balance with the space overhead in practice.
We also note that one could pair a partial extractor with multiple port functions as in the 103-qubit system, if there is symmetry in the space of logical operators. 
This technique will boost the measurement capacity of a partial extractor at little cost.

Now consider multiple non-overlapping partial extractors on the same code bridged together. When the bridges are activated, we can measure product operators as in Corollary~\ref{cor:bridged_partial_extractors}. When the bridges are deactivated, we can perform parallel measurements of non-overlapping operators.
An interesting subcase of parallel measurement is to simultaneously measure multiple non-overlapping representatives of the same logical operator.
As noted in Remark~5 of~\cite{williamson2024low}, by measuring $m$ equivalent representatives in parallel and taking majority vote on the measurement result, we only need $O(d/m)$ rounds of syndrome measurement for to maintain fault distance $d$, instead of $O(d)$ rounds. 

We further note that many ideas discussed in this section are especially applicable to hypergraph product codes~\cite{tillich2013quantum}, and potentially lifted/balanced product codes~\cite{panteleev2021degenerate,breuckmann2021balanced,panteleev2022good}.
A key feature of hypergraph product codes is that their logical operators have representatives with disjoint supports and well-understood structure.\footnote{See Ref.~\cite{eczoo_hypergraph_product} for a comprehensive yet non-exhaustive list of references on hypergraph product codes. Appendix B of Ref.~\cite{quintavalle2022reshape} describes such a basis of operators.}
Therefore, it is highly plausible to exploit such structures and construct partial or full extractors and port functions that have low space overhead. 
We defer these studies to future works.

The key motivation behind introducing the partial extractors in this paper is to illustrate that there is a full range of options between a full single-block extractor and measurement graphs supported on individual logical operators. 
We conclude that the construction and optimization of extractors on practically relevant codes is a promising direction with vast room for exploration.

\section{QLDPC Architecture for Pauli-Based Computation}\label{sec:architectures}

With all the required machinery in place, we now discuss how to build a fault-tolerant, fixed-connectivity computational architecture with EAC blocks. 
Our architecture natively supports logical measurements of great flexibility; consequently, when combined with sources of high quality magic states, we can realize universal quantum computation through Pauli-based computation (PBC)~\cite{bravyi2016trading}. 
As there are a plethora of proposals for magic state factories, both practical and asymptotic, in this work we focus on the PBC architecture and leave the choice of magic state factory as a user-defined variable. 

In Section~\ref{sec:main_architecture} we construct our \textbf{extractor architecture}, and discuss how a magic state factory can be integrated. In Section~\ref{sec:compilation} we discuss how to compile and compute with this architecture. As we illustrate, similar to the popular work Ref.~\cite{litinski2019game}, most of the Clifford gates in a circuit can be compiled away, which means the computation would be dominated by $\pi/8$ rotations. 
However, unlike the simplified case in Ref.~\cite{litinski2019game} where each $\pi/8$ rotation acts on the entire computer, our architecture consists of many EAC blocks, and we allow $\pi/8$ rotations to act simultaneously on each EAC block. This means computation on our architecture can be much more parallel.

\subsection{Architecture}\label{sec:main_architecture}

Let $\qcode$ be a $[[n,k,d]]$ quantum LDPC code, and let $\fullsystem$ be an EAC block. An extractor architecture based on $\fullsystem$ can be specified by a constant degree graph $\map = (\blocks, \bridges)$, which we call the \textbf{block map}. 
Every vertex in $\blocks$ corresponds to an EAC block, and every edge in $\bridges$ corresponds to a bridge system connecting two EAC blocks. 

In more detail, let us enumerate the EAC blocks as $\system{1}, \cdots, \system{B}$. For every edge $e = (i,j)\in \bridges$, we connect the extractors $\EXT_i, \EXT_j$ with a bridge system as in Lemma~\ref{lem:bridging_extractors}. 
We denote the full architecture as $\archi = (\fullsystem, \map)$. 
Let $B = |\blocks|$ denote the number of blocks, and $R = |\bridges|$ denote the number of bridges. 
Since $\map$ is a constant degree graph, we let $R = \arcDegree B$. 

\begin{figure}[H]
    \centering
    \includegraphics[width = 0.9\textwidth]{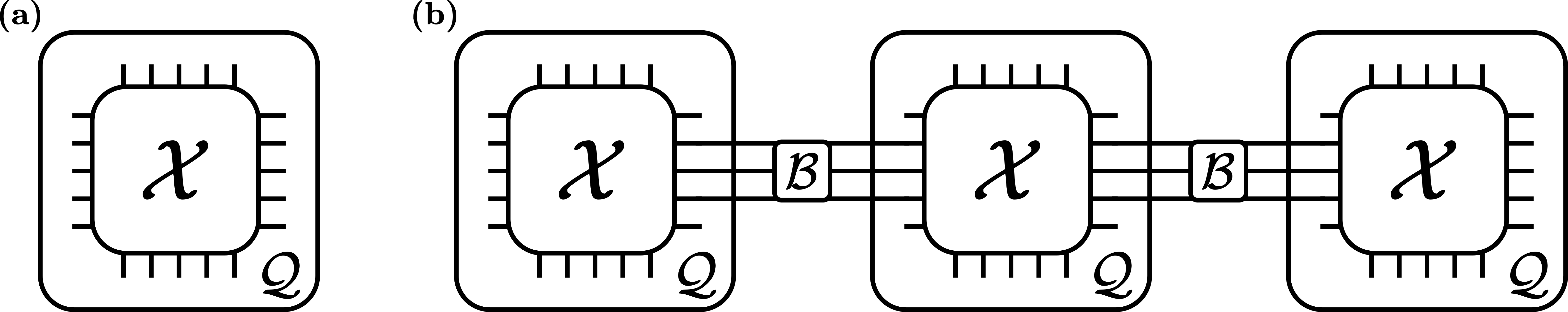}
    \caption{Abstraction of the extractor architecture. 
    \textbf{(a)} A high-level depiction of an EAC block. 
    \textbf{(b)} Sketch of a simple extractor architecture $\archi$, where the block map $\map$ is a 3-vertex line graph. The EAC blocks are connected by bridge systems of qubits and checks.
    }
    \label{fig:architecture_1_copy}
\end{figure}

$\archi$ has the following computational and structural properties. 
\begin{proposition}[Architecture Parameters]\label{prop:parameters}
    $\archi$ has the following parameters.
    \begin{itemize}[itemsep = 0pt]
        \item The maximum degree of connectivity of qubits (data or check) in $\archi$ is $O(1)$. 
        
        \item Suppose an EAC block uses $\ratio n$ many physical (data and check) qubits. Then $\archi$ has a total of $B(\ratio n + \alpha (2d-1))$ qubits, where $\ratio Bn$ comes from the EAC blocks and $\alpha B(2d-1)$ comes from the bridges.
        
        \item $\archi$ operates on a total of $Bk$ many logical qubits. 
        For later compilation purposes, we allocate one logical qubit in every EAC block to be a logical ancilla, so the full workspace size is $K:= B(k-1)$. 
        At memory mode, when only the code blocks $\qcode_1, \cdots, \qcode_B$ are active, the logical qubits are protected by code distance $d$. 
    \end{itemize}
\end{proposition}
\begin{proof}
    The statements follow directly from our analysis in Section~\ref{sec:extractors}. Particularly, since $\map$ is constant degree, the extractor system in every EAC block is connected by at most $O(1)$ many bridges. 
    Every bridge consists of $d$ data qubits and $d-1$ checks, which contributes less than $\alpha B(2d-1)$ qubits in total. If one further cellulates the cycle checks introduced by bridges, then the qubit count increases accordingly.
\end{proof}
\begin{remark}
    Let $r = k/n$ denote the rate of $\qcode$. For high rate QLDPC codes, especially those we wish to implement in practice, the multiplicative space overhead of $\archi$ scales as 
    \begin{align}
        \frac{B(\ratio n + \alpha (2d-1))}{K} \approx \frac{\ratio}{r} + 2\alpha\frac{d}{k}.
    \end{align}
    For asymptotically good codes where $d, k = \Theta(n)$, the scaling becomes $\ratio + O(1)$.
    From Lemma~\ref{lem:single_block_extractor} we see that $\ratio\le O((\log n)^3)$, while in practice we expect $\ratio$ to be a small constant when optimized on medium-scale QLDPC codes.
\end{remark}

\begin{remark}
    An important feature of $\archi$ is that the amount of global connectivity required to implement $\archi$ is potentially a very small fraction of the overall system size. 
    For instance, we can take $\map$ to be a one-dimensional line or a two-dimensional grid. 
    The choice of $\map$ impacts compilation choices, which we discuss in Section~\ref{sec:compilation}. 
    Nonetheless, so long as $\map$ is a connected graph, $\archi$ can implement an arbitrary PBC circuit on all the logical qubits (when supplied with magic states).
    We further note that as shown in Ref.~\cite{Tremblay2022}, any LDPC architecture can be implemented into a multi-planar layout, where all qubits are embedded into a two-dimensional lattice and connections are partitioned into planar subsets. 
\end{remark}

We now consider the computational capacity of $\archi$. 
From here on we use the term ``one logical cycle'' to denote $O(d)$ rounds of syndrome measurement.\footnote{For simplicity of discussions, we do not differentiate between syndrome measurement cycles of different QLDPC codes.} 
From Section~\ref{sec:extractors}, it is clear that if $\map$ is connected, then any logical Pauli operator supported on the $K$ logical qubits of $\archi$ can be measured in one logical cycle. 
This is, however, not parallel. 
Instead, we consider partitions of $\map$ into connected subgraphs. 

\begin{definition}[Subgraph Partition]\label{def:subgraph_partition}
    For a graph $G = (V, E)$, a \textit{connected subgraph partition} of $G$ is a partition of $V$ into $V = S_1\sqcup\cdots \sqcup S_p$ such that every set of vertices $S_i$ induces a connected subgraph in $G$.
\end{definition}

\begin{proposition}[Computational Capacity]\label{prop:measurements}
    Let $\blocks = \BS_1\sqcup \cdots \sqcup \BS_p$ be a connected subgraph partition of $\map$. Let $\MO = \{O_1, \cdots, O_p\}$ be a set of Pauli operators such that $O_i$ is supported on the logical qubits in the code blocks in $\BS_i$, namely $\cup_{j\in \BS_i}\qcode_j$.
    Then the operators in $\MO$ can be measured on $\archi$ in parallel in one logical cycle.
\end{proposition}
\begin{proof}
    As illustrated in Section~\ref{sec:extractors}, different measurements on a extractor system can be performed by activating and deactivating different set of edges/qubits.
    To measure operators in $\MO$, for each $\BS_i$, we find a spanning tree $\BT_i$ of the subgraph that $\BS_i$ induces in $\map$. 
    We activate all bridge edges (and cycle checks) in the trees $\BT_1\cdots, \BT_p$, and deactivate all other bridge edges (and cycle checks).
    Now the extractors of the EAC blocks in every set $\BS_i$ is connected by bridges, and by repeated application of Lemma~\ref{lem:bridging_extractors} we see that they form an extractor for the joint codespace $\cup_{j\in \BS_i}\qcode_j$. Therefore we can measure $O_i$ in one logical cycle. Measurements of different $O_i$ can be done in parallel since they are supported on different code blocks. 
\end{proof}

For our compilations, we simply partition $\blocks$ into non-overlapping pairs of vertices connected by an edge in $\bridges$. 
One could also consider partitions of the map $\map$ into vertex-disjoint paths or edge-disjoint paths, which are well studied combinatorial problems and have been utilized for compilation on surface code architectures~\cite{Beverland2022edgedisjoint}. Note that by grouping EAC blocks together or simply choosing larger EAC blocks, we can compile more Clifford gates away at the expense of less parallel magic state teleportation.

To perform universal quantum computation with PBC, the logical measurements we execute need to be supported on magic states as well. As discussed earlier, we leave the choice of magic state factory to the user, and note that there are plenty of options, from practical schemes to asymptotic proposals.

\begin{remark}[Source of magic states]\label{rmk:magicstates}
Efficient generation of high fidelity magic states, especially $\ket{T} = \ket{0} + e^{i\pi/4}\ket{1}$ states, has been a key focus in the study of fault-tolerant quantum computation. 
We briefly list a few schemes that one might consider using in pair with our architecture. 

For practical purposes, conventional distillation schemes~\cite{Knill2004fault,bravyi2005magic,Meier2013magic,Bravyi2012magic} have been studied and optimized in many proposals of magic states factories (see, for instance, Refs.~\cite{Jones2013multilevel,OGorman2017magic,litinski2019magic,Gidney2019efficientmagicstate}, and an experimental demonstration in Ref.~\cite{rodriguez2024experimental}). 
While multi-round distillation is often considered too costly, there has been steady improvements in post-selected magic state injection techniques~\cite{Chamberland2019faulttolerantmagic,Chamberland2020verylow,Bombin2024,Hirano2024zerolevel}, which are not scalable asymptotically but are often cheaper than distillation in practically relevant regimes.
A leading proposal is magic state cultivation~\cite{gidney2024magic}, which injects a $\ket{T}$ state into a surface code of distance~15. 
Using cultivation directly, or loading cultivated $\ket{T}$ state into one round of 15-to-1 or 10-to-2 distillation, will make suitable single-output factories to pair with our architecture.

In the asymptotic regime, following a line of works on distillation~\cite{Bravyi2012magic,Haah2018codesprotocols,Krishna2019towards},
there has been a recent breakthrough in achieving constant space overhead magic state distillation~\cite{wills2024constant} using asymptotically good quantum codes with transversal non-Clifford gates~\cite{nguyen2024goodbinaryquantumcodes,golowich2024asymptoticallygoodquantumcodes}. 
A concurrent distillation procedure proposed in Ref.~\cite{nguyen2024quantum} achieves low spacetime overhead. 
These constructions provide batch-output factories, as they produce a large number of magic states at once.

Another factory proposal that does not require distillation is that of a magic state fountain (a phrase borrowed from Ref.~\cite{zhu2025topological}), where we use a high distance code that supports transversal non-Clifford gates to directly produce logical magic states. A line of recent works~\cite{zhu2023gates,scruby2024quantum,breuckmann2024cups,golowich2024quantum,lin2024transversal,zhu2025topological} has constructed or proposed methods to construct QLDPC codes with transversal or constant-depth non-Clifford gates. 
These codes are not asymptotically good, but they have (close to) constant rate, which is important for a cascading fountain. 
The recent work Ref.~\cite{jain2024high} proposed a collection of more practically relevant codes with transversal $T$ gates, which have better parameters than the three dimensional topological codes~\cite{Bombin2007exact,kubica2015universal,Vasmer2019three}. 
Improving the parameters and performances of these codes to a practically competitive level is an important challenge.
\end{remark}

\begin{figure}[t]
    \centering
    \includegraphics[width = 0.6\textwidth]{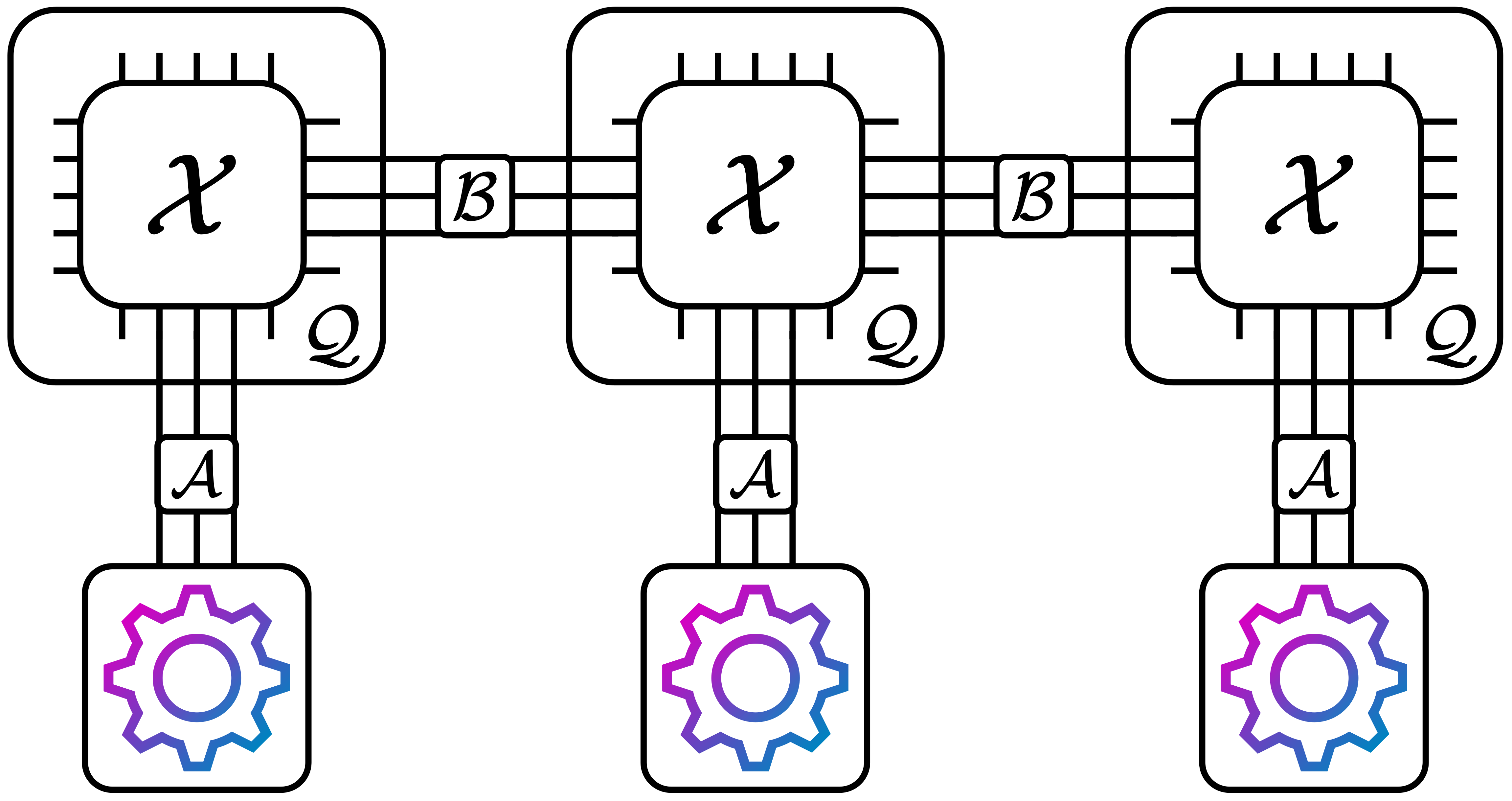}
    \caption{High level depiction of an extractor architecture paired with local factories (colored gear boxes). The factories are connected to the extractors through adapters $\mathcal{A}$. These local factories could be powered by, for instance, magic state cultivation.
    }
    \label{fig:architecture_local_fact}
\end{figure}

\begin{figure}[t]
    \centering
    \includegraphics[width = 0.6\textwidth]{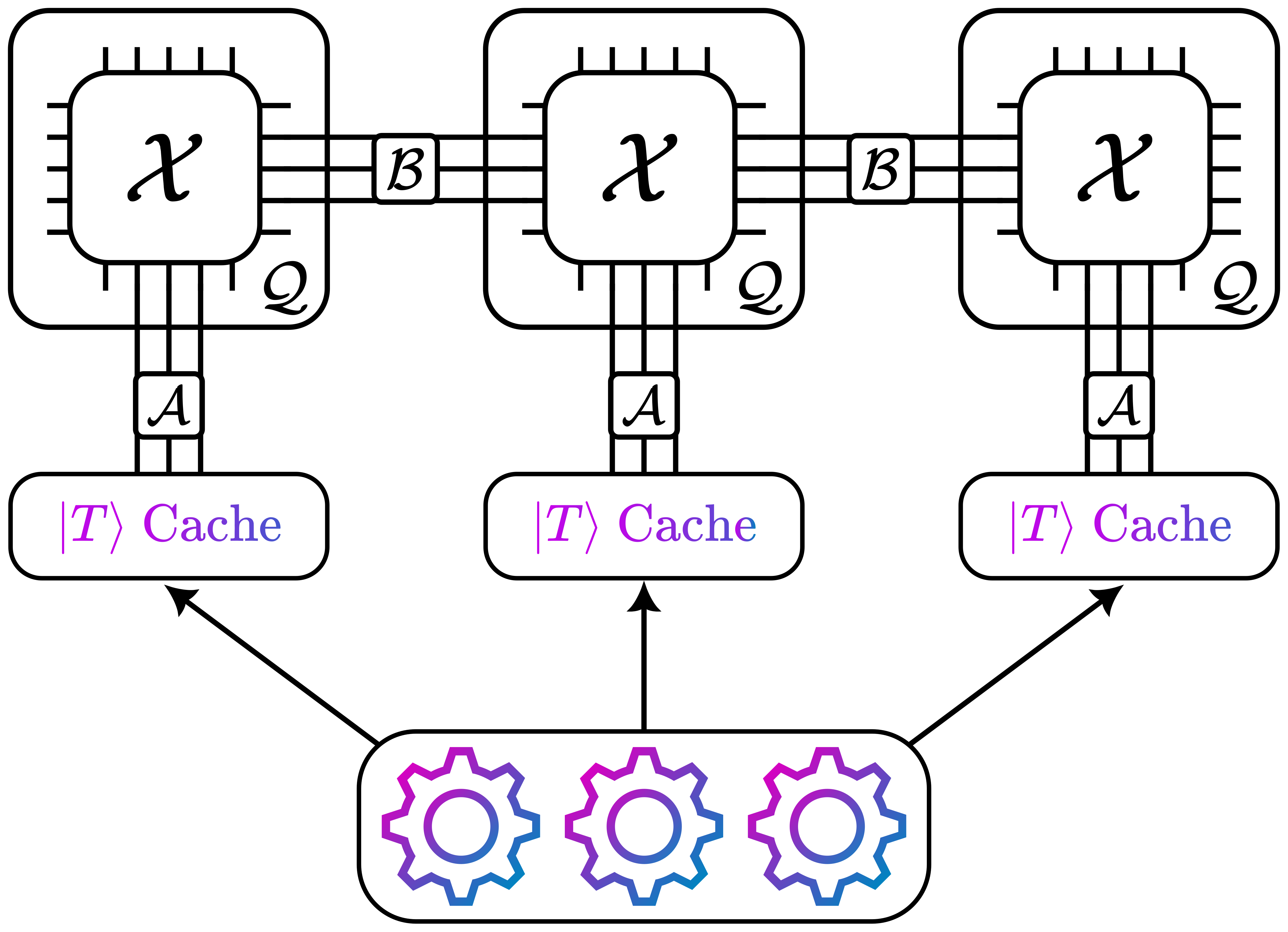}
    \caption{High level depiction of an extractor architecture paired with a global factory. 
    Every EAC block (or every cluster of EAC blocks) should have a cache of magic states, supplied by the global factory.
    The caches are connected to the extractors through adapters $\mathcal{A}$.
    If the caches are themselves high-rate QLDPC memories, it may make sense to also equip them with extractors (not drawn) to facilitate the storage and consumption of magic states.
    }
    \label{fig:architecture_global_fact_copy}
\end{figure}

For optimal functionality of the extractor architecture and for simplicity of analysis, we would like every EAC block to have a steady supply of magic states. 
There are two natural models in which this requirement can be satisfied. 
In the \textit{local factory} model, we can connect a single-output factory (such as cultivation
or cultivation combined with one round of distillation) to every EAC block using an adapter. Similarly, we can group a constant number of physically nearby EAC blocks into a cluster, and supply each cluster with a single-output factory using multiple adapters.
As distillation and post-selected injection schemes continue to be optimized and improved, the local factory model may be promising to realize on small-to-medium scale hardwares.
In the \textit{global factory} model, we can build a batch-output factory that produces a large number of magic states at once, and shuttle these magic states to individual EAC blocks using adapters and potentially intermediate transit memory.
Importantly, in either model, the factory and auxiliary components can be built with any code of one's choice, as the adapter construction is capable of connecting arbitrary quantum code families together and generating entanglement between them. 
Consequently, magic states produced in one code family can be used for non-Clifford gates in an EAC block based on a different code family.

We note that in the bicycle architecture of Ref.~\cite{yoder2025tour}, an EAC block is a bivariate bicycle code memory augmented by a partial extractor (which are called logical processing units), and the block map $\map$ is a line. A source of magic state, based on either distillation or cultivation, is connected to one end of the line.  

With these discussions in mind, we now return our focus to the extractor architectures and discuss compilation, assuming every EAC block has a steady supply of $\ket{T}$ states.

\subsection{Compilation}\label{sec:compilation}

To illustrate computation in the extractor architecture, consider a circuit $\C$ on $K = B(k-1)$ qubits made of the following gates: Paulis, $H, S, \CNOT, $ and $\T$.
Fix a partition $\pt$ of qubits which allocate the logical qubits into the $B$ EAC blocks. 
With respect to this partition, we say that a $\CNOT$ gate is \textit{in-block} or \textit{single-block} if both of its target qubits belong to the same EAC block, and \textit{cross-block} if they belong to different EAC blocks.
Since the block map $\map$ is user-defined, in this work we make the following simplifying assumption on $\C$ and $\pt$.

\begin{definition}[Compatibility]\label{def:Compatibility}
    A circuit $\C$ is said to be \textit{compatible} with a partition $\pt$ if every cross-block $\CNOT$ gate in $\C$ acts on two EAC blocks that are connected by a bridge (equivalently, an edge in $\bridges$).
\end{definition}

In this way we leave the choice of $\map$, $\pt$, and the compilation of compatible circuit $\C$ to the user.
As noted above, as long as $\map$ is a connected graph, any circuit $\C$ can be compiled (with $\mathsf{SWAP}$ gates, for instance) to fit with any partition $\pt$.\footnote{This is similar to the qubit routing problem~\cite{cowtan2019routing,paler2017routing}, but with codeblocks rather than individual qubits.}
Later we discuss how the choice of $\pt$ affects time overhead. 
We further assume that $\C$ ends with a round of standard $Z$ basis measurements on all logical qubits.
This last assumption is both standard and simplifies some of our discussions.

Our compilation owes inspiration to, and borrows notation from, the work of Litinski~\cite{litinski2019game}.
We begin by converting $\T$ gates in the circuit $\C$ into $\zt$ rotations (Figure~\ref{fig:compilation}a), and cross-block $\CNOT$ gates into $\zx$ rotations (Figure~\ref{fig:compilation}b). 
Since all other gates are single-block Cliffords, they can all be exchanged (Figure~\ref{fig:compilation}c) to the end of the circuit, turning a $\zt$ rotation into a single-block $\pi/8$ rotation and a $\zx$ rotation into a cross-block $\pi/4$ rotation.
As we assumed that the circuit $\C$ ends with a round of standard $Z$ basis measurements, the single-block Cliffords can be absorbed into these measurements, turning each single-qubit $Z$ measurement into a single-block Pauli measurement.
In this manner, single-block Clifford gates in the circuit are essentially ``free''.
After this step, we see that the circuit is now composed of rounds of single- and cross-block rotations, with $k-1$ rounds of single-block Pauli measurements at the end.

To implement these single- and cross-block rotations, we compile them into Pauli measurements as in Figure~\ref{fig:compilation}d, e. 
A cross-block $\pi/4$ rotation requires an ancilla qubit, which we have allocated in every EAC block. 
Therefore, this rotation can be implemented with three measurements: one to initialize the ancilla state, one cross-block measurement, and one single-qubit measurement on the ancilla qubit. 
Note that at runtime, after we perform the single-qubit measurement on the ancilla qubit, there is no need to reinitialize the ancilla qubit to $\ket{0}$ as we can simply change the Pauli basis on this ancilla. 
Therefore every cross-block rotation (and equivalently, $\CNOT$ gate) can be implemented with two Pauli measurements.

\begin{figure}[t]
    \centering
    \includegraphics[width = \textwidth]{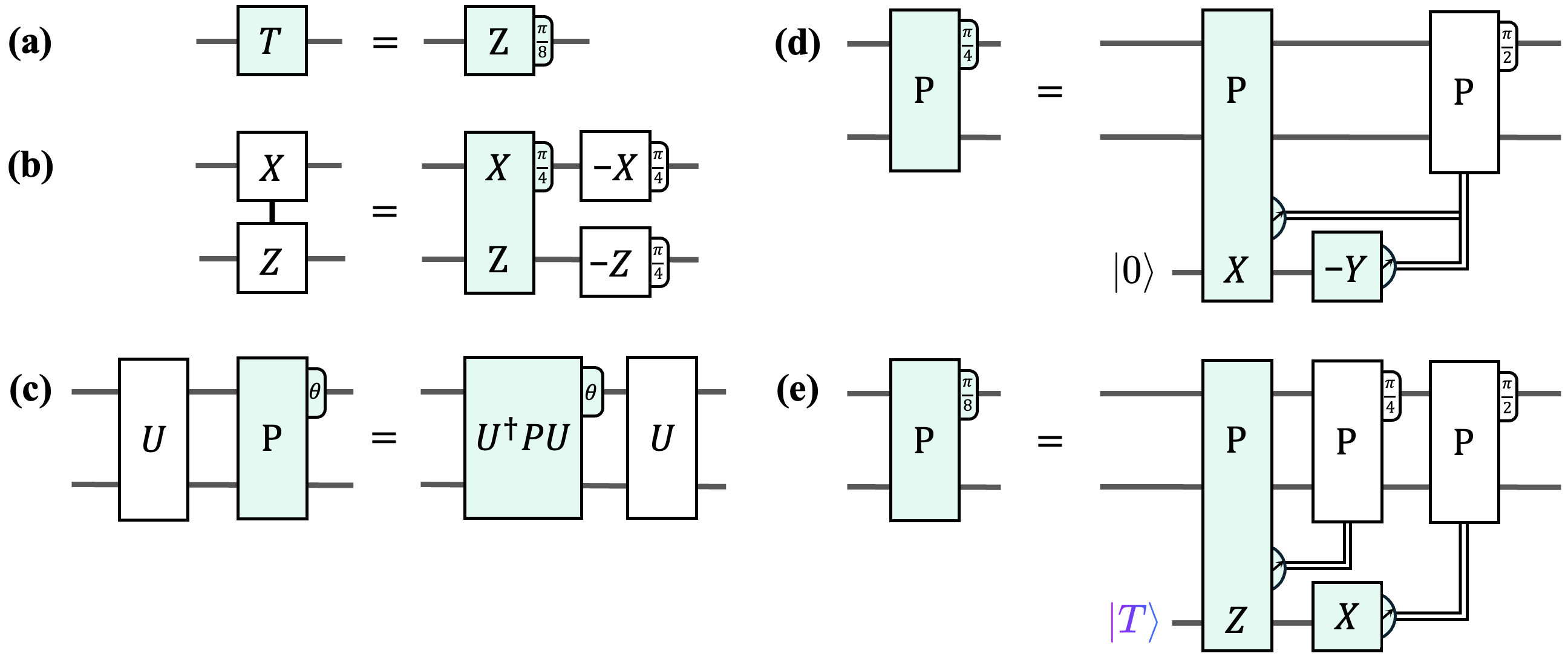}
    \caption{Compilation identities we use. These circuit diagrams largely follow those of~\cite{litinski2019game}. Colored operations are those that we will compile and later execute on EAC blocks, uncolored operations are in-block Clifford gates that will be absorbed into the final measurements or Pauli gates, which can be tracked in the logical Pauli frame.
    \textbf{(a)} A $\T$ gate is a $\zt$ rotation.
    \textbf{(b)} A $\CNOT$ gate can be written as a $\zx$ rotation, with two single-qubit Clifford rotations.
    \textbf{(c)} A Pauli rotation of degree $\theta$ around $P$ can be exchanged with a Clifford operation $U$, turning into a Pauli rotation of degree $\theta$ around $U^{\dagger}PU$. 
    \textbf{(d)} A Pauli $\pi/4$ rotation can be compiled into two Pauli measurements. This uses an ancilla qubit and the measurement results control a Pauli correction.
    \textbf{(e)} A Pauli $\pi/8$ rotation can be compiled into two Pauli measurements, using a $\ket{T}$ magic state. The two measurements control a Clifford and a Pauli correction.}
    \label{fig:compilation}
\end{figure}

A single-block $\pi/8$ rotation can be implemented by a joint Pauli measurement on an EAC block and a $\ket{T}$ magic state, followed by a single-qubit measurement on the magic state. 
At runtime, this second measurement can be scheduled in parallel with the next measurement on the EAC block, since a Pauli correction can be tracked in the Pauli frame.
Therefore every $\T$ gate really takes one logical cycle.
A related technique is the auto-corrected $\pi/8$ rotations (see Figure~17 of Ref.~\cite{litinski2019game}), where the first measurement (joint between EAC block and magic state) no longer induces a Clifford correction, which means we no longer need to wait for decoding and Clifford feedforward to finish to proceed with our next measurement.
These ideas can further reduce the time overhead of our logical computations.

One of the key features of an extractor, or an EAC block, is that logical measurement of a Pauli operator with heavy logical and/or physical support is as easy as logical measurement of a Pauli operator with low logical and/or physical support -- we simply activate and deactivate different coupling edges $\leftrightarrows$. This is in stark contrast to the case of lattice surgery on the surface code, as considered in Ref.~\cite{litinski2019game}, where measuring a Pauli operator supported on more logical qubits would require surgery using larger ancilla patches.
Note that we have a slight complication -- some of our measurements are supported on magic states as well. 
Nonetheless, these measurements are easy to perform in one logical cycle. 
If the magic states are stored in, for instance, surface codes, the adapter from the EAC block can be directly connected to a logical operator on the surface code. 
If the magic states are stored in more compact memories, we can build (partial) extractors on these memories as well. 
Therefore, every single- and cross-block measurement we obtain through the above compilation costs one logical cycle. 

How much can we parallelize these measurements?
We define the following notion of reduced depth.
\begin{definition}[Reduced Depth]
    For a Clifford plus $\T$ circuit and a partition $\pt$, let $\newc$ be the sub-circuit of $C$ we obtain by removing all in-block Clifford gates. If $\newc$ has depth $\depth$, we say that $C$ has reduced depth $\depth$.
\end{definition} 
Evidently, this is a scheduling problem on $\newc$ where we need to serialize $\newc$ such that at any timestep, every block is acted on by at most one $\T$ or $\CNOT$ gate. 
We give a worst case upper bound of $O(k\depth)$ on the depth of the serialized circuit.

\begin{definition}[Edge Coloring of Graphs]
    Given a graph $G = (V, E)$, an edge coloring of $G$ is a natural number $c$ and a function $f: E\rightarrow [c]$ such that for any two edges $e_1,e_2$ with overlapping endpoints, $f(e_1)\ne f(e_2)$.
\end{definition}

\begin{theorem}[Vizing's theorem~\cite{vizing1965critical,Berge1991vizing}]\label{thm:vizing}
A multigraph with maximum degree $\delta$ and maximum multiplicity $\mu$ can be edge-colored with $\delta + \mu$ colors. 
\end{theorem}

\begin{lemma}[Serialization]\label{lem:serialize}
    Let $A$ be a depth-1 circuit of one- and two-qubit gates, acting on $B$ blocks of qubits, each of size $k-1$. Then we can serialize the gates in $A$ into a circuit $A'$ of depth at most $2(k-1)$, such that at any timestep, every block is acted on by at most one gate.
\end{lemma}
\begin{proof}
    We first schedule all the cross-block gates.
    We define a scheduling graph $G = (V, E)$ with one vertex per block. For every cross-block 2-qubit gate acting on the $i,j$th blocks, we add an edge $(i,j)$ to $G$. 
    Then $G$ is a multi-graph with degree at most $k-1$ and multiplicity at most $k-1$. 
    By Vizing's theorem, $G$ can be colored with at most $2(k-1)$ colors.
    Let $f$ be such a coloring, and let $E_c$ denote all the edges with color $c\in [2(k-1)]$ in $G$. 
    Then the gates corresponding to edges in $E_c$ can be scheduled in one timestep. 
    We can therefore schedule all cross-block gates in $2(k-1)$ depth.
    For all the in-block gates, note that in the $2(k-1)$ rounds of cross-block gates, every block must be idle for at least $k-1$ rounds (since every block has at most $k-1$ gates in $A$).
    Therefore, all in-block gates can be scheduled in the same $2(k-1)$ rounds as well. 
\end{proof}

\begin{remark}\label{rmk:serial_worstcase}
    It is important to note that this worst case upper bound is proved assuming that a block partition has been fixed.
    We prove this lemma simply to provide a theoretical upper bound. 
    In practice, given a circuit $C$ and a block map $\map$, we get to optimize the partition $\pt$ and circuit $C$ to minimize the serialized depth.
    Therefore, we believe this serialization lemma does not capture the time overhead in practice.
\end{remark}

Lemma~\ref{lem:serialize} gives us an upper bound on the depth of $\compiledC$.
Recall that $\depth$ is the reduced depth of the input circuit $C$.
Since every rotation can be implemented with two logical measurements, 
we have that 
\begin{equation}\label{eq:time_analysis}
    \text{Depth of $\compiledC$} < 4k\cdot \depth + k,
\end{equation}
where the last summand of $k$ is for the final Pauli measurements at the end of the circuit.
The number of logical cycles needed to execute all measurements in $\compiledC$ is therefore $4k\cdot \depth + k$. 
However, many of these measurements will be supported on magic states. 
In specific architectures, the time cost of supplying these magic states to the EAC blocks needs to be carefully accounted for.

\begin{remark}[Magic state supply cost]\label{rmk:magic_supply}
    The space and time overhead of supplying magic states to EAC blocks depends on the choice of factories and their various parameters, including throughput, batch size and cache size.
    We consider a few simplified special cases. 
    Let $\Tmagic$ denote the average number of logical cycles needed to supply every EAC block in $\archi$ with one $\ket{T}$ state. 
    If the local magic state caches are small (such as local magic state cultivation patches), then the number of logical cycles needed to execute $\compiledC$ can be upper bounded as $O(k\depth\cdot \max(1,\Tmagic))$.
    In particular, if $\Tmagic$ is constant, which means the magic states are produced fast enough (by global or local factories), the time cost can be bounded as $O(k\depth)$.
    Alternatively, if we have large local magic state caches (such as high-rate QLDPC codes), then the factories can continuously produce magic states to be stored in these caches.
    In this setting, the time cost will be impacted by the number of times an EAC block requests upon an empty cache and have to wait for the magic state to arrive. 
    If the supply and storage are powerful enough such that the caches are never empty, then the logical cycle count can be bounded by $4k\depth + k$, as in equation~\eqref{eq:time_analysis}.
    Otherwise, if there are $\deptht$ calls to empty caches throughout the execution of $\compiledC$, the logical cycle count would be bounded by $4k\depth + \deptht\cdot \Tmagic + k$.
    The real cost in practice depends on the choice of specific factories and caches.
\end{remark}

As in Remark~\ref{rmk:serial_worstcase}, we emphasize that equation~\eqref{eq:time_analysis} is a loose upper bound that does not account for the choice of $C$ and $\pt$.
We do observe, however, that in general increasing the block size $k$ will reduce the throughput requirement on the magic state factory, while potentially increasing the serialization depth.
We leave more accurate resource estimation of specific algorithms, such as factoring~\cite{Shor}, with optimized choices of $C, \qcode, \map$ and $\pt$, to future works. 
Notably, Ref.~\cite{yoder2025tour} provided estimations on the bicycle architecture for tasks such as simulation of the time evolution of a two-dimensional transverse-field Ising model.

Throughout this work, we have taken one logical cycle to be $O(d)$ physical syndrome measurement cycles. 
It has been shown that certain families of QLDPC codes can be single-shot decoded~\cite{fawzi2020constant,gu2024single,nguyen2024quantum}, which means they have decoders that only use a single round of syndrome information to output corrections with logical error rates of $e^{-O(d)}$.\footnote{There are several similar yet distinct definitions of single-shot QEC. We refer readers to Refs.~\cite{Bombin2015single,campbell2019theory} for further references.}
In other words, a logical cycle on these QLDPC memories is $O(1)$ syndrome measurement cycles.
A related recent work~\cite{hillmann2024single} has shown that lattice surgery can be single-shot in higher dimensional topological codes, which are known to have single-shot decoders as well~\cite{Kubica2022single,bombin2015gauge}. 
Therefore, a natural and important question is: can extractor measurements be single-shot on specific QLDPC code families, and if so, at what costs?
We look forward to exploring this in future works.

\subsection{Architecture: Practical Considerations}\label{sec:archi_practical}

Our extractor architecture opens a plethora of practical questions to be considered. We sample and discuss a few of them here. 

In Section~\ref{sec:partial_extractors}, we proposed that on QLDPC codes with constant depth Clifford gates such as automorphism gates, a partial extractor can be aided by these gates to reach full measurement capacity. 
Evidently, this approach increases the depth of physical gates we need to perform for one logical measurement. 
It is therefore crucial to optimize the space-time tradeoff, where using a bigger partial extractor or more partial extractors can speed up the synthesis of logical measurements.
Another relevant idea is whether the techniques of algorithmic fault-tolerance~\cite{zhou2024algorithmic} can be applied to this combination of constant-depth gates and logical measurements to reduce the overhead of a logical cycle.
We leave these studies to future works.

As elaborated in the previous section, the choice of magic state factory, $k, \map$ and $\pt$ all have significant impact on the overall overhead of this architecture. 
For fixed design hardwares such as superconducting systems, we would often be in the situation where the factory, $k$ and $\map$ have been built, and we have an algorithm to execute. 
The compilation problem is then to optimize the circuit $C$ and partition $\pt$ for the lowest time overhead, see equation~\eqref{eq:time_analysis}.
For more flexible hardware systems such as neutral atom devices, the code, extractor, factory and block map $\map$ can all be designed based on the algorithm, and the logical computation can be aided with other primitives such as transversal block-wise $\mathsf{CNOT}$s.

Interestingly, having transversal block-wise $\mathsf{CNOT}$s as a primitive allows us to reduce the \textit{online} time overhead at the cost of \textit{offline} time overhead and additional space overhead. 
Here offline refers to a pre-processing stage before we execute an algorithm, where we can prepare ancilla code states and store them as error corrected memory. Online refers to the real-time execution of the algorithm, where we can utilize the ancillary states we prepared. 
The main idea behind this online/offline tradeoff can be summarized as follows. 
While code surgery and extractors enable flexible logical measurements, at the moment fault-tolerance of the protocols require $O(d)$ syndrome measurement cycles. 
In comparison, Steane's~\cite{Steane1997active} and Knill's~\cite{Knill2005realistic} syndrome measurement scheme can perform logical measurements with one round of transversal $\mathsf{CNOT}$s, utilizing ancilla code blocks prepared in certain logical stabilizer states. 
Similarly, homomorphic measurement~\cite{huang2023homomorphic,xu2024fast} utilizes ancilla code blocks to implement flexible logical measurements on homological product codes with one round of transversal $\mathsf{CNOT}$s.
These ancilla code blocks can be prepared offline using EAC blocks or an extractor architecture, while prior to this work many logical stabilizer states are prepared through distillation~\cite{Zheng2018efficient}.
This approach effectively reduces the online time overhead of a logical measurement to $O(1)$, at the cost of space overhead.

A further extension of this idea has interesting consequences. 
Suppose we have a family of QLDPC codes $\tilde{Q}$ with constant-depth and addressable non-Clifford gates. 
We can use EAC blocks as stabilizer state factories, which produce resource states that when teleported induces arbitrary Clifford circuits on $\tilde{Q}$. 
An universal circuit can then be executed as follows: perform non-Clifford gates in low-overhead on a block of $\tilde{Q}$, and teleport in segments of Clifford circuits between rounds of non-Clifford gates.
Since any logical Pauli measurement can be done in one logical cycle on an EAC block, high fidelity preparation of any stabilizer state can be done in $O(k)$ logical cycles.
The fault-tolerant computation therefore alternates between constant-depth non-Clifford gates and Clifford circuit teleportation, where the resource states are prepared offline or in parallel.
We note that there does not yet exist any QLDPC code $\tilde{Q}$ with practically relevant parameters, and  constant-depth, addressable non-Clifford gates. 
This is an exciting direction for future inquiry~\cite{he2025quantum,he2025asymptotically,guemard2025good}.

\newpage

\bibliographystyle{unsrt}
\bibliography{main}

\begin{appendices}

\section{Omitted Proofs}\label{apdx:proofs}

\cordecongestion*
\begin{proof}
    Recall that $\MR$ is ordered as $\MR = \{C_1, \cdots, C_{|E|-|V|+1}\}$, where every cycle $C_i$ overlaps with at most $\log_2(|V|)\cdot \rho$ cycles later in the ordering.
    We compute a partition greedily in \textit{reverse} order. 
    \begin{algorithm}[H]
        \caption{Greedy Partition}
        \label{alg:greedy_partition}
        \begin{algorithmic}[1] 
        \REQUIRE An ordered collection of cycles $\MR$.
        \ENSURE A partition $\MR = \bigcup_{i=1}^t \MR_i$ such that each $\MR_i$ contains non-overlapping cycles.
        \STATE Initialize $i = 1$.
        \WHILE{$\MR$ is non-empty}
            \STATE Enumerate $\MR = \{C_1, \cdots, C_r\}$ in order, initialize $\MR_i = \varnothing$. Set $j = r$.
            \WHILE{$j > 0$}
                \STATE If $C_j$ overlaps with any cycle in $\MR_i$, set $j = j-1$. Otherwise, add $C_j$ to $\MR_i$ and set $j = j-1$.
            \ENDWHILE
            \STATE Remove cycles in $\MR_i$ from $\MR$ and set $i = i+1$.
        \ENDWHILE
        \RETURN Computed partition $\{\MR_i\}$.
        \end{algorithmic}
    \end{algorithm}

    It is straightforward to see that each $\MR_i$ output by the above algorithm contains non-overlapping cycles.
    It suffices for us to argue that $t\le \log_2(|V|)\cdot \rho+1$. 
    We prove inductively that after $i$ partitions are computed, every cycle $C$ remaining in $\MR$ overlaps with at most $\log_2(|V|)\cdot \rho - i$ many later cycles in $\MR$. 
    The base case with $i = 0$ is our assumption. 
    Inductively, the computed partition at step $i+1$ is $\MR_{i+1}$. 
    For every cycle $C_j$ in $\MR\setminus \MR_{i+1}$, there must exists $C_{j'}\in \MR_{i+1}$, $j' > j$ such that $C_j$ overlaps with $C_{j'}$, because otherwise $C_j$ would be included in $\MR_{i+1}$. 
    Therefore, the amount of later cycles in $\MR$ overlapping with $C_j$ must decrease by at least $1$ after step $i+1$.
    This implies that the algorithm terminates after $ \log_2(|V|)\cdot \rho+1$ steps. 
\end{proof}

\thickenbasis*
\begin{proof}
    Let $\MV, \ME$ denote the vertices and edges of $G\square J_\ell$. 
    For this proof, we identify cycles with their indicator vectors in $\FF_2^{|\ME|}$.
    We make a few observations.
    First, the cycles in $\MT$ are linearly independent. 
    Moreover, for any cycle $C$ of $G$, any $r\in [\ell]$, we have
    \begin{equation}
        \{C\times \{r'\}:r'\in [\ell]\}\subset C\times \{r\} + \MT. 
    \end{equation}
    In other words, every cycle of $G$ has an equivalent representative on every level of $G\square J_\ell$. 
    Furthermore, $\MT\cup \MR$ is linearly independent.
    We now count the dimension of the cycle space of $G\square J_\ell$.
    From Definition~\ref{def:thicken}, we see that $|\ME| = |E|\times \ell + |V|\times (\ell-1)$, and $|\MV| = |V|\times \ell$.
    Therefore the cycle space has rank $|\ME| - |\MV| + 1 = |E|\times (\ell-1) + |E| - |V| + 1 = |\MT\cup \MR|$. 
    Therefore $\MT\cup \MR$ is a cycle basis, and for every cycle in $\MR$ we can choose any of its representations on different levels to measure in the basis. 
\end{proof}

\thickenlemma*
\begin{proof}
    For a set of vertices $U\subseteq V\times [\ell]$, let $U_j$ denote $U\cap (V\times \{j\})$ for $0\le j\le \ell$. 
    Let $u_j\in \FF_2^{|V|}$ be the indicator vector of $U_j$, and let $M$ be the incidence matrix of $G$. 
    Then $u_j^\top M$ is indicator vector of the edge boundary of $U_j$.
    By the structure of the thickened graph, we have
    \begin{align}
        |\delta_{G\square J_\ell} U| 
        &= \sum_{j=2}^\ell |u_{j-1} + u_{j}| + \sum_{j=1}^{\ell} |u_j^\top M|.
    \end{align}
    In the equation above, the first term counts the edges between levels of vertices $V\times \{j-1\}, V\times \{j\}$, and the second term counts the edges within each level of edges $E\times \{j\}$.
    Citing the relative Cheeger constant of $G$, we have
    \begin{align}
        |\delta_{G\square J_\ell} U| 
        &\ge \sum_{j=2}^\ell |u_{j-1} + u_{j}| + \sum_{j=1}^{\ell} \beta \min(t, |U_j\cap (P\times \{j\})|, |(P\times \{j\}) \setminus U_j|).
    \end{align}
    Suppose $\ell\beta \le 1$. Then,
    \begin{align}
        |\delta_{G\square J_\ell} U| 
        &\ge \beta\ell\sum_{j=2}^\ell |u_{j-1} + u_{j}| + \sum_{j=1}^{\ell} \beta \min(t, |U_j\cap (P\times \{j\})|, |(P\times \{j\}) \setminus U_j|), \\
        &\ge \beta\sum_{j=1}^\ell \left( \min(t, |U_j\cap (P\times \{j\})|, |(P\times \{j\}) \setminus U_j|) + \sum_{i=2}^\ell |u_{i-1} + u_{i}| \right).
    \end{align}
    Note that for all $j, r\in [\ell]$, we have by the triangular inequality
    \begin{align}
        &\phantom{\ge} \min(t, |U_j\cap (P\times \{j\})|, |(P\times \{j\}) \setminus U_j|) + \sum_{i=2}^\ell |u_{i-1} + u_{i}| \\
        &\ge \min(t, |U_j\cap (P\times \{j\})|, |(P\times \{j\}) \setminus U_j|) + |u_j + u_r| \\
        &\ge \min(t, |U_r\cap (P\times \{r\})|, |(P\times \{r\}) \setminus U_r|). \label{eq:triangle_inequality}
    \end{align}
    Therefore $|\delta_{G\square J_\ell} U| \ge \beta\ell\min(t, |U_r\cap (P\times \{r\})|, |(P\times \{r\}) \setminus U_r|)$, as desired. 
    If instead $\ell\beta > 1$, let 
    \begin{equation}
        j^* = \argmin_{j\in [\ell]} \min(t, |U_j\cap (P\times \{j\})|, |(P\times \{j\}) \setminus U_j|).
    \end{equation}
    Then we have
    \begin{align}
        |\delta_{G\square J_\ell} U| 
        &\ge \sum_{j=2}^\ell |u_{j-1} + u_{j}| + \sum_{j=1}^{\ell} \beta \min(t, |U_j\cap (P\times \{j\})|, |(P\times \{j\}) \setminus U_j|) \\
        &\ge \sum_{j=2}^\ell |u_{j-1} + u_{j}| + \ell\beta \min(t, |U_{j^*}\cap (P\times \{{j^*}\})|, |(P\times \{{j^*}\}) \setminus U_{j^*}|) \\
        &\ge \sum_{j=2}^\ell |u_{j-1} + u_{j}| + \min(t, |U_{j^*}\cap (P\times \{{j^*}\})|, |(P\times \{{j^*}\}) \setminus U_{j^*}|) \\
        &\ge \min(t, |U_r\cap (P\times \{r\})|, |(P\times \{r\}) \setminus U_r|),
    \end{align}
    where the last equation follow from Equation~\eqref{eq:triangle_inequality}.
\end{proof}

\BridgingLemma*

\begin{proof}
To construct a bridge $B$ of $d$ edges, we wish to apply Lemma~\ref{lem:sparse_bridge} to $P_1, P_2$. 
However, from the desiderata (Theorem~\ref{thm:graph_desiderata}) it is not guaranteed that $P_1, P_2$ induce connected subgraphs in $G_1, G_2$ (even though in all constructions we presented in this paper the induced subgraphs are connected). 
We deal with this minor technicality by a detailed analysis of the connectedness of the subgraphs.

Recall that a logical operator $\ML$ supported on a set of qubits $L$ is irreducible if the restriction of $\ML$ to any proper subset of qubits in $L$ is not a logical operator or a stabilizer.
Consider the operator $\ML_1$, and let $\ML_1'$ be a restriction of $\ML_1$ that is irreducible (if $\ML_1$ is irreducible already, then $\ML_1 = \ML_1'$).
Let $L_1'$ be the support of $\ML_1'$. 
For simplicity of presentation we assume that $\ML_1'$ is a $Z$-operator, as the same argument works regardless of the Pauli components of $\ML_1'$. 

We now construct a set of helper edges, $H_1$, for our argument. Let $P_1' = f_1(L_1')$, consider the set of $X$-stabilizers, $S_1, \cdots, S_m$, which have support on $L_1'$.
Let $K_i$ denote the qubits in both $S_i$ and $\ML_1'$. 
From desideratum~\ref{des:path_matching}, there is a path matching $\mu_i$ of size $O(1)$ in $G_1$ of $f_1(K_i)\subseteq P_1'$. 
For each $\mu_i$, if vertices $u, v\in f_1(K_i)$ is connected by a path in $\mu_i$, we add an edge $(u,v)$ to the set of helper edges $H_1$. 

We now prove that the helper edges form a connected graph on the vertices $P_1'$. More precisely, for any two vertices $u, v\in P_1'$, there exists a path of edges in $H_1$ that connects $u, v$. 
The first observation we make is that since $L_1'$ is irreducible, the sets $K_i$, when interpreted as indicator vectors (of length $|L_1'|$) over $\FF_2$, generate the parity check matrix of the repetition code (equivalently, generate the space of all even weight vectors).
Therefore, for any two vertices $f_1^{-1}(u), f_1^{-1}(v)\in L_1'$, there exists a collection of $K_i$'s such that their $\FF_2$ sum is $\{u, v\}$. 
Translating to the helper edges $H_1$, this means that we have a collection of edges $H_{u,v}\subseteq H_1$ such that in the graph $A_{u,v} = (P_1', H_{u,v})$, all vertices have even degree except for $u, v$, which has odd degree.
A simple fact from graph theory is that in such a graph, there must be a path between the two vertices of odd degree. 
This proves that $A_1 = (P_1',H_1)$ is a connected graph.

We can apply the same arguments to $G_2, \ML_2$ and obtain a graph $A_2 = (P_2', H_2)$ with the same properties. 
Let $w$ denote the maximum length of any path in the path matchings of desiderata~\ref{des:path_matching} of $G_1, G_2$.
We see that if there is a length $p$ path between $u, v$ in $A_1$ or $A_2$, then there is a length at most $pw$ path between $u, v$ in $G_1$ or $G_2$, respectively. 
We can now apply Lemma~\ref{lem:sparse_bridge} to the graphs $A_1, A_2$, which gives us a bridge of $d$ edges that we can use to connect $A_1, A_2$ into $A$, as well as $G_1, G_2$ into $G$. The cycles created in $A$ has length at most $8$, which means the length of cycles created in $G$ is at most $8w \in O(1)$. 
Moreover, each edge in $A$ is used at most twice.
Suppose every edge in $G$ is used in at most $\delta$ many path matchings in desiderata~\ref{des:path_matching}, then these new cycles in $G$ has congestion at most $2\delta \in O(1)$. 
This means $G$ satisfy desiderata~\ref{des:cycle_basis} again.

The rest of the desiderata are simple.
Clearly, $G$ is connected and $G$ has constant degree. 
For desideratum~\ref{des:path_matching}, for any stabilizer $S$ of $\qcode$, let $K(S, \ML_1\ML_2)\subseteq L_1\cup L_2$ denote the set of qubits in $Q$ on which $S$ and $\ML_1\ML_2$ anticommute. 
Since $L_1\cap L_2 = \varnothing$, $K(S, \ML_1\ML_2) = K(S, \ML_1)\cup K(S, \ML_2)$, which means the path matchings $\mu(S, \ML_1) \cup \mu(S, \ML_2)$ together gives a path matching $\mu(S, \ML_1\ML_2)$ for $K(S, \ML_1\ML_2)$. 
These new path matchings $\mu(S, \ML_1\ML_2)$ are again sparse.
Finally, $\beta_d(G, P_1\cup P_2)\ge 1$ by Lemma~\ref{lem:bridged_expansion}, which completes our proof.
\end{proof}

\BridgingExtractors*
\begin{proof}
    We follow the proof of Lemma~\ref{lem:bridge} to construct a bridge $B$ of $d$ edges between $P_1,P_2$, and the resulting graph $X$ satisfy desideratum~\ref{ext-des:cycle_basis}.
    One can easily check that $\EXTG, F$ satisfy desideratum~\ref{ext-des:connected},~\ref{ext-des:degree}, and~\ref{ext-des:relative_expansion} of Definition~\ref{def:extractor_desiderata}. For~\ref{ext-des:path_matching}, we note that the codes $\qcode_1,\qcode_2$ are supported on disjoint qubit sets, which means the union of two collections $\mathcal{E}_1, \mathcal{E}_2$ satisfy~\ref{ext-des:path_matching}. This completes our proof.
\end{proof}

\end{appendices}

\end{document}

%% file: preamble.tex
\usepackage[margin=1in]{geometry}

\usepackage[]{graphicx}
\usepackage{amssymb}
\usepackage{enumitem}
\usepackage{appendix} 
\usepackage{newtxtext}
\usepackage{amsmath}
\usepackage{amsthm}
\usepackage{amsfonts}
\usepackage{comment}
\usepackage{color}
\usepackage{stmaryrd}
\usepackage{braket}
\usepackage{bbm}
\usepackage{thmtools} 
\usepackage{thm-restate}
\usepackage{rotating}
\usepackage{qcircuit}
\usepackage{xcolor}
\usepackage{algorithm}
\usepackage{algorithmic}
\floatname{algorithm}{Algorithm} %

\usepackage{titlesec}

\titleformat{\section}{\large\bfseries}{\thesection}{1em}{}
\titleformat{\subsection}{\normalsize\bfseries}{\thesubsection}{1em}{}
\titleformat{\subsubsection}{\normalsize\bfseries}{\thesubsubsection}{1em}{}

\def\thesection{\arabic{section}}
\def\thesubsection{\arabic{section}.\arabic{subsection}}
\def\thesubsubsection{\arabic{section}.\arabic{subsection}.\arabic{subsubsection}}
\makeatletter
\renewcommand{\p@subsection}{}
\makeatother

\theoremstyle{definition}
\newtheorem{theorem}{Theorem}
\newtheorem{corollary}[theorem]{Corollary}

\newtheorem{proposition}[theorem]{Proposition}

\newtheorem{definition}[theorem]{Definition}
\newtheorem{remark}[theorem]{Remark}

\newtheorem{lemma}[theorem]{Lemma}

\newcommand{\MA}{\mathcal{A}}
\newcommand{\MB}{\mathcal{B}}
\newcommand{\ME}{\mathcal{E}}

\newcommand{\ML}{\mathcal{L}}
\newcommand{\MO}{\mathcal{O}}

\newcommand{\MQ}{\mathcal{Q}}
\newcommand{\MR}{R}
\newcommand{\MS}{\mathcal{S}}
\newcommand{\MT}{\mathcal{T}}
\newcommand{\MV}{\mathcal{V}}

\newcommand{\qcode}{\mathcal{Q}}

\newcommand{\FF}{\mathbb{F}}
\DeclareMathOperator{\supp}{supp}
\DeclareMathOperator{\im}{im}
\DeclareMathOperator*{\argmin}{arg\,min}

\usepackage{mathrsfs,dsfont}

\newcommand{\EXTG}{X}
\newcommand{\EXT}{\mathcal{X}}
\newcommand{\fullsystem}{\qcode \leftrightarrows \EXT}
\newcommand{\system}[1]{\qcode_{#1} \leftrightarrows \EXT_{#1}}
\newcommand{\partialsystem}[1]{\qcode \leftrightarrows \EXT_{#1}}
\newcommand{\archi}{\mathbb{A}}
\newcommand{\ratio}{\lambda}
\newcommand{\arcDegree}{\alpha}
\newcommand{\BS}{\mathbb{S}}
\newcommand{\BT}{\mathbb{T}}

\newcommand{\map}{\mathbb{M}}
\newcommand{\blocks}{\mathbb{V}}
\newcommand{\bridges}{\mathbb{E}}

\newcommand{\pt}{\Pi}
\newcommand{\depth}{\Lambda}
\newcommand{\T}{T}

\newcommand{\CNOT}{\mathsf{CNOT}}
\newcommand{\deptht}{\depth_{\T}}

\newcommand{\newc}{\tilde{C}}
\newcommand{\compiledC}{C_{\text{comp}}}

\newcommand{\C}{C}

\newcommand{\zt}{Z_{\pi/8}}

\newcommand{\zx}{(Z\otimes X)_{\pi/4}}
\newcommand{\Tmagic}{T_{\text{magic}}}

\usepackage{hyperref} %

\usepackage{etoolbox}

\makeatletter
\renewcommand{\thebibliography}[1]{%
  \section*{\refname}%
  \small%
  \list{\@biblabel{\alph{enumiv}}}%
    {
    \settowidth\labelwidth{\@biblabel{#1}}%
     \leftmargin\labelwidth
      \advance\leftmargin\labelsep
    }%
    \renewcommand{\makelabel}[1]{##1\hfill} %
  }
\makeatother